%% file: ieee_standalone.tex
\documentclass[10pt,twocolumn,twoside]{ieee/IEEEtran}
\usepackage{xcolor,soul,framed} 

\colorlet{shadecolor}{yellow}
\usepackage[pdftex]{graphicx}
\graphicspath{{figures/}}
\DeclareGraphicsExtensions{.eps,.pdf,.jpeg,.png}

\usepackage[cmex10]{amsmath}
\usepackage{array}
\usepackage{mdwmath}
\usepackage{mdwtab}
\usepackage{eqparbox}
\usepackage{url}
\usepackage{amssymb}
\usepackage{amsthm}
\usepackage{amsfonts}
\usepackage{hyperref}
\usepackage{xurl}
\usepackage{orcidlink}

\newtheorem{theorem}{Theorem}[section]
\newtheorem{lemma}[theorem]{Lemma}
\newtheorem{corollary}[theorem]{Corollary}
\newtheorem{definition}[theorem]{Definition}

\newcommand{\cfg}[1]{\hyperref[tab:configurations]{Config.~#1}}

\newif\ifnotIEEE
\notIEEEfalse
\newif\ifwithConfergenceAnalysis
\withConfergenceAnalysisfalse

\begin{document}
    \title{Distributed Model Predictive Control with Adaptive Safety Zones for Multi-Fleet Drone Operations}
    \author{
        \IEEEauthorblockN{Linda M\"{u}mken\,\orcidlink{0009-0005-2646-9397}\IEEEauthorrefmark{1},
            Diyar Altinses,\orcidlink{0009-0005-7928-5874}\IEEEauthorrefmark{1},
            Michael Schwung,\orcidlink{0000-0002-8420-737X}\IEEEauthorrefmark{2},
            Stefan Lier,\orcidlink{0000-0002-3314-7610}\IEEEauthorrefmark{3},
            Andreas Schwung,\orcidlink{0000-0001-8405-0977}\IEEEauthorrefmark{1}\\}
        \IEEEauthorblockA{\IEEEauthorrefmark{1}Department of Automation Technology
            South Westphalia University of Applied Sciences,
            Soest, Germany\\
            \{muemken.linda, altinses.diyar, schwung.andreas\}@fh-swf.de\\}
        \IEEEauthorblockA{\IEEEauthorrefmark{2}Institute of Automation and Computer Control,
            Ruhr University Bochum, 44801 Bochum, Germany 
            Michael.schwung@rub.de\\}
        \IEEEauthorblockA{\IEEEauthorrefmark{3}Department of Logistics and Supply Chain Management,
            South Westphalia University of Applied Sciences, Meschede, Germany 
            lier.stefan@fh-swf.de}}



\maketitle

    \begin{abstract}
        The increasing deployment of autonomous drones in restricted airspaces demands safe, collision-free operations at maximum vehicle density. We address two key challenges: determining the geometric capacity limits of a given airspace volume and ensuring stable, collision-free trajectory optimization up to those limits. Our approach employs model predictive control (MPC) with adaptive, speed-dependent safety zones. Each drone's safety radius scales with its braking distance. It is tight at low speeds and expanded at high speeds. We prove that this adaptive scheme achieves feasibility up to the geometric packing limit evaluated at the minimum radius, maximizing airspace utilization where fixed-radius methods fail. Within this framework, we develop and compare centralized and distributed MPC formulations. Centralized MPC provides faster convergence with full state information, while distributed MPC (DMPC) enables each drone to optimize locally based on detected neighbors, accommodating mixed fleets with non-cooperative agents. We establish Lyapunov stability under sufficient conditions on the adaptation parameter, drone density, and prediction horizon, and derive modified geometric capacity bounds showing that sphere-packing density limits remain valid at the minimum radius. Simulations confirm that the adaptive framework achieves feasibility in scenarios where fixed-radius approaches become infeasible.
    \end{abstract}
    \input{extras/ntp}


    \IEEEpeerreviewmaketitle


    \section{Introduction}

    \IEEEPARstart{A}{utonomous} drone swarms are increasingly deployed in space-constrained environments such as warehouses, infrastructure inspections, and urban delivery corridors~\cite{faa_utm_conops,kopardekar2016utm,schwung_drone_logistics}, where multiple drones must share limited airspace safely and efficiently~\cite{yasin2020collision,huang2019collision}. A central challenge is determining how many drones a bounded volume can accommodate while guaranteeing collision-free operation~\cite{bauranov2023capacity,doole2020estimation}.

    Existing approaches~\cite{drones10020139} rely on fixed safety zones, which must account for the worst-case velocity and therefore wastes airspace in congested scenarios. A hovering drone requires far less separation than one traveling at maximum speed, yet a constant radius cannot exploit this. Furthermore, centralized coordination is often unavailable due to the absence of shared communication infrastructure~\cite{kuwata2007distributed,icao_utm_framework}, and each drone must rely on broadcast neighbor states subject to delays and limited update rates. Model predictive control (MPC) is well-suited to trajectory optimization under safety and kinematic constraints~\cite{MAYNE2000789,BORRELLi}, but replacing fixed radii with velocity-dependent ones introduces new analytical challenges: the optimizer may steer drones toward configurations that become infeasible upon subsequent acceleration. Establishing feasibility and stability of MPC under such conditions is therefore non-trivial.

    We address these limitations by replacing the fixed radius with an adaptive, velocity-proportional safety sphere that reflects braking distance, allowing slower drones to maintain smaller safety zones and pass through bottlenecks that fixed-radius methods block. Based on these sphere types, we develop both a centralized MPC that optimizes all trajectories jointly and a distributed MPC in which each drone solves a local problem using kinematic information from neighboring drones~\cite{shorinwa2023dmpc,luis2020dmpc_drones,GRAFE202279}, treating all agents uniformly regardless of fleet membership. To make the distributed problem tractable, all drones broadcast their kinematic state; sensing uncertainty can be incorporated by inflating the adaptive radii for estimation errors. Within this framework, we directly connect feasibility to airspace utilization. This means that if the MPC problem remains feasible for a given drone configuration, the configuration can operate safely up to the geometric packing limit. Building on~\cite{drones10020139}, we show that the capacity bounds remain valid and reformulate jamming conditions for velocity-dependent safety zones.

    The principal contributions of this work are as follows:
    \begin{itemize}
        \item We develop a distributed MPC framework with adaptive, velocity-dependent safety zones for mixed-traffic airspace, treating all nearby agents uniformly without requiring coordination or fleet membership information.
        \item We prove feasibility and Lyapunov stability for the centralized MPC and extend both to the distributed setting. The extension quantifies the coupling error from distributed iteration and establishes a contraction condition that preserves the centralized stability margins.
        \item We show that the fundamental capacity limits from~\cite{drones10020139} remain valid under the adaptive formulation, deriving explicit bounds on the critical drone density and the minimum prediction horizon required for feasibility.
        \item We validate the framework in simulation, comparing centralized and distributed MPC with static and adaptive spheres. The results demonstrate feasibility where fixed-radius methods fail and quantify computational trade-offs.
    \end{itemize}

    The remainder of this paper is organized as follows. Section~\ref{sec:related_work} reviews related work on distributed control, safety zones, and airspace capacity. Section~\ref{sec:preliminaries} introduces the adaptive safety radius, agent dynamics, and communication topology, followed by the centralized and distributed MPC formulations in Section~\ref{sec:framework}. The theoretical analysis in Section~\ref{sec:theory} develops throughput bounds for narrow passages, feasibility and stability proofs, and geometric capacity limits for the adaptive formulation. Section~\ref{sec:evaluation} validates the framework through simulation and illustrates the trade-offs between centralized and distributed approaches, before Section~\ref{sec:conclusion} concludes with limitations and future directions.


    \section{Related Work}\label{sec:related_work}

    This section reviews relevant literature from three domains: collision avoidance in urban airspace, safety zone formulations, and distributed model predictive control.

    \subsection{Collision Avoidance in Urban Drone Operations}

    Urban environments combine severe spatial constraints with high traffic density, where buildings, restricted zones, and dynamic obstacles create confined corridors~\cite{bauranov2023capacity, doole2020estimation}. Unlike commercial aircraft, drones must navigate these complex three-dimensional environments in real time~\cite{yasin2020collision, huang2019collision}. MPC-based path planning approaches address this by decomposing the problem into global planning for static environments and local collision resolution for dynamic threats~\cite{baca2018mpc, lindqvist2020nmpc, dentler2016realtime, olcay2024dynamic}, with multi-drone coordination remaining an active research area~\cite{ramezani2023uav}. Detect-and-avoid systems have further enabled beyond-visual-line-of-sight operations~\cite{mohajerin2023sensors}, though swarm scenarios intensify the tracking challenge significantly~\cite{doukhi2025sim}. Regulatory frameworks now mandate specific separation criteria for drone operations in controlled airspace~\cite{faa_utm_conops, icao_utm_framework, kopardekar2016utm, bauranov2023capacity}.

    Unlike these approaches, which focus on real-time avoidance or regulatory compliance, this work provides formal geometric capacity guarantees. We derive packing bounds and feasibility proofs to determine the maximum number of drones a given airspace can accommodate, analyzing the regime near the packing boundary where adding a single drone can make the MPC problem infeasible.

    \subsection{Safety Zones and Trajectory Planning}

    Collision avoidance is typically formalized using spherical safety zones with a fixed radius around each agent~\cite{zhang2018survey}, which enables capacity analysis via sphere-packing theory~\cite{hales2005proof, conway1999sphere, RevModPhys.82.2633} and simplifies optimization-based constraint formulation. The scoring algorithm framework in~\cite{schwung_drone_logistics, drones10020139} classifies drone logistics decisions into static, dynamic, and semi-static categories. Safety zone sizing falls into the semi-static category, where the initial radius depends on drone characteristics but can adapt to environmental conditions. However, fixed radii ignore the fact that stopping distance scales quadratically with velocity~\cite{fiorini1998velocity, vandenberg2011rvo}, which motivates velocity-dependent formulations that adjust the safety zone according to the instantaneous kinematic state~\cite{wang2025velocityapf, sun2023adaptive}. For trajectory planning, Bézier curves offer smooth, efficiently computable paths with inherent continuity guarantees~\cite{schwung_bezier_trajectory}. In hierarchical schemes, stand-on agents move freely, and give-way agents plan around communicated trajectories. This introduces uncertainty when stand-on agents deviate~\cite{richter2016polynomial}.

    Rather than fixing the safety radius, this work extends the geometric packing bounds of~\cite{drones10020139} by replacing the constant radius with an adaptive formulation tied to physical braking distance.

    \subsection{Distributed Model Predictive Control}

    In shared airspace, drones from different operators must coexist without a common coordination framework~\cite{faa_utm_conops, icao_utm_framework}. Distributed MPC addresses this by enabling each agent to solve local optimization problems based on observations of nearby drones~\cite{kuwata2007distributed, tallamraju2018decentralized}, with successful applications ranging from event-triggered replanning~\cite{GRAFE202279} and online trajectory generation~\cite{luis2020dmpc_drones} to cooperative formation control~\cite{drones9050366}. Two fundamental update paradigms exist: Jacobi-type schemes, where all agents optimize simultaneously using previous-iteration data~\cite{stomberg2022dmpc_compendium}, and Gauss--Seidel-type schemes, where agents optimize sequentially using the most recent solutions from predecessors~\cite{koehler2022sequential_dmpc}. The Gauss--Seidel approach naturally aligns with stand-on/give-way protocols for collision avoidance~\cite{schwung_bezier_trajectory} and handles external non-cooperative agents by treating their trajectories as exogenous constraints~\cite{yoshikawa2023chance, mesbah2016stochastic}. Stability in such mixed environments can be ensured via control barrier functions or chance constraints~\cite{ames2017cbf, zeng2021mpccbf, thirugnanam2024gcbf}.

    Although adaptive safety zones~\cite{sun2023adaptive}, distributed MPC~\cite{luis2020dmpc_drones, shorinwa2023dmpc}, and sphere packing theory~\cite{hales2005proof, conway1999sphere} are well studied individually, no prior work combines all three within a framework supporting heterogeneous, non-cooperative agents.
    This work provides exactly this combination: sphere packing bounds determine the geometric capacity limit at which the MPC problem becomes infeasible, while adaptive safety zones and distributed coordination maximize utilization up to that limit, with formal feasibility and stability guarantees for both centralized and distributed formulations.


    \section{Preliminaries}\label{sec:preliminaries}

    This section defines the adaptive safety radius, the drone dynamics, and the communication topology underlying the MPC formulations in Section~\ref{sec:framework}.

    \subsection{Adaptive Safety Radius Definition}

    We derive the adaptive safety zone directly from kinematic considerations. During maximum deceleration, the safety zone must be large enough to bring the drone to a complete stop from its current velocity.

    \begin{definition}[Adaptive Safety Radius]\label{def:adaptive_radius}
        Each drone $i$ at time $t$ occupies a safety zone represented by a closed sphere $B_{r_i(t)}(\mathbf{p}_i)$. The adaptive radius is defined by
        \begin{equation}
            r_i(t) = r_i(\mathbf{v}_i(t)) = r_{\min} + \alpha \cdot s_{\text{st}}(\mathbf{v}_i(t)),
        \end{equation}
        where $r_{\min} > 0$ is the base radius accounting for the physical dimensions of the drone and sensor uncertainties, $\alpha \in (0, 1]$ is a scaling parameter that modulates the influence of velocity in the safety zone and $s_{\text{st}}(\mathbf{v}_i(t))$ is the stopping distance. We write $r_i(t)$ as shorthand for the composite map $r_i(\mathbf{v}_i(t))$ whenever the velocity dependence is clear from context.
    \end{definition}

    Intuitively, faster drones need larger safety zones due to longer braking distances. The parameter $\alpha$ trades off conservative safety ($\alpha$ close to 1) against more aggressive space utilization ($\alpha$ close to 0).

    The pairwise safety constraint becomes time-varying and state-dependent, requiring that
    \begin{equation}\label{eq:pairwise_separation}
        \|\mathbf{p}_i(t) - \mathbf{p}_j(t)\| \geq r_i(t) + r_j(t), \quad \forall i \neq j, \forall t.
    \end{equation}

    \subsection{Agent Model and Communication Topology}\label{subsec:agent_model}

    To integrate the adaptive safety radius into an optimization-based control scheme, we formalize the drone dynamics and the communication structure regulating information exchange among agents.

    \subsubsection*{Modeling Framework}

    Let $\Omega \subset \mathbb{R}^3$ be a bounded, connected, Lipschitz domain representing the airspace. There are $N$ identical drones indexed by $i \in \{1, \dots, N\}$. Each drone has a goal $\bar{\mathbf{p}}_i \in \Omega$, occupies a safety zone represented by the closed sphere $B_r(\mathbf{p}_i)$ of radius $r$ centered at $\mathbf{p}_i$, and is subject to boundary conditions $B_r(\mathbf{p}_i(t)) \subset \Omega$ for all $i, t$. Safety requires pairwise separation $\|\mathbf{p}_i(t) - \mathbf{p}_j(t)\| \geq 2r$ for all $i \neq j$ and for all $t$.

    In a physical quadrotor, translational acceleration is generated indirectly. Rotor speeds $\omega_k$ produce a total thrust $T = c_T \sum_k \omega_k^2$, where $c_T$ is a rotor-specific thrust coefficient, as well as body torques~$\boldsymbol{\tau}$. A cascaded control architecture~\cite{schwung_bezier_trajectory} uses a fast inner-loop attitude controller whose bandwidth far exceeds the MPC sampling rate. Under standard timescale separation~\cite[Chapter~11]{khalil_nonlinear_2002}, the translational dynamics reduce to the double-integrator model

    \begin{align}
        \dot{\mathbf{p}}_i &= \mathbf{v}_i, \\
        \dot{\mathbf{v}}_i &= \mathbf{u}_i,
    \end{align}
    with state $\mathbf{x}_i = (\mathbf{p}_i, \mathbf{v}_i) \in \mathbb{R}^3 \times \mathbb{R}^3$, subject to hard control and speed bounds $\|\mathbf{u}_i(t)\| \leq U_{\max}$ and $\|\mathbf{v}_i(t)\| \leq V_{\max}$. A single scalar parameter $U_{\max} = T_{\max}/m - g$ encapsulates the physical propulsion limits, rotor saturation, motor dynamics, and maximum tilt angle. The feasibility and stability results derived in Section~\ref{sec:theory} under this abstraction can be shown to hold for the full nonlinear rigid-body quadrotor model as well, since the inner-loop controller reduces the translational dynamics to a perturbed double integrator with bounded disturbance that is absorbed by the same Lyapunov mechanism used in the stability proof.

    \subsubsection*{Neighborhood Graph}

    The neighbor set $\mathcal{N}_i$ of drone $i$ determines which agents are considered in the local collision avoidance constraints. Two natural choices are the following.
    \begin{itemize}
        \item \textit{Communication-radius coupling:} $\mathcal{N}_i = \{j : \|\mathbf{p}_i - \mathbf{p}_j\| \leq r_{\text{com}}\}$ for a given communication radius $r_{\text{com}} > 0$. This reflects practical communication range limitations.
        \item \textit{Full connectivity:} $\mathcal{N}_i = \{1, \ldots, N\} \setminus \{i\}$. This recovers the collision avoidance scope of the centralized formulation.
    \end{itemize}
    We assume an ideal network in which each drone receives its neighbors' positions at every MPC step and uses their predicted trajectories for local optimization. Safety is preserved as long as $r_{\text{com}}$ is large enough to include all drones whose safety zones could potentially overlap.


    \section{Problem Formulation}\label{sec:framework}

    Given the adaptive safety radius (Definition~\ref{def:adaptive_radius}), the double-integrator dynamics (Section~\ref{subsec:agent_model}), and the pairwise separation constraint~\eqref{eq:pairwise_separation}, the central problem of this work can be stated as follows.

    \noindent\textbf{Problem.}
    Consider $N$ drones in a bounded airspace $\Omega \subset \mathbb{R}^3$, each subject to acceleration and velocity bounds $U_{\max}$ and $V_{\max}$, with velocity-dependent safety radii $r_i(\mathbf{v}_i)$ per Definition~\ref{def:adaptive_radius}. Design a control law $\mathbf{u}_i(t)$ for each drone such that:
    \begin{enumerate}
        \item[\textbf{(F)}] \emph{Feasibility:} The pairwise separation $\|\mathbf{p}_i(t) - \mathbf{p}_j(t)\| \geq r_i(t) + r_j(t)$ and the airspace constraint $B_{r_i(t)}(\mathbf{p}_i(t)) \subset \Omega$ hold for all $i \neq j$ and $t \geq 0$.
        \item[\textbf{(S)}] \emph{Stability:} Each drone converges to its goal, $\|\mathbf{p}_i(t) - \bar{\mathbf{p}}_i\| \to 0$ as $t \to \infty$.
        \item[\textbf{(D)}] \emph{Distributed operation:} Without centralized coordination, drone $i$'s control law depends only on its own state and the communicated states of neighbors $j \in \mathcal{N}_i$.
        \item[\textbf{(C)}] \emph{Capacity:} Determine the maximum number $N_{\mathrm{crit}}$ for which \textbf{(F)} and \textbf{(S)} remain jointly solvable in $\Omega$.
    \end{enumerate}

    We address this problem through MPC. The speed dependence of the safety radius allows drones to shrink their safety zones at low speeds, enabling passage through bottlenecks that fixed-radius approaches block. Safety is maintained throughout, since the reduced radius remains proportional to the actual braking distance. We first present the centralized formulation, then reformulate it as a distributed MPC.

    \subsection{Centralized MPC}\label{subsec:centralized_mpc}
    With sampling time $\Delta t > 0$ and horizon $H \in \mathbb{N}$, let $s \triangleq k + h$ for $h \in \{0, \ldots, H-1\}$.
    At each decision time $k\Delta t$, the centralized MPC problem is formalized as:

    \begin{equation}\label{eq:centralized_mpc}
        \begin{aligned}
            \min_{\{\mathbf{u}_i(s)\}} &\sum_{h=0}^{H-1} \sum_{i=1}^{N} \left[\ell(\mathbf{x}_i(s), \mathbf{u}_i(s)) + \lambda\|\mathbf{v}_i(s)\|^2\right]\\
            \text{s.t.} \quad &\mathbf{p}_i(s+1) = \mathbf{p}_i(s) + \Delta t \cdot \mathbf{v}_i(s) + \tfrac{1}{2}\Delta t^2 \cdot \mathbf{u}_i(s),\\
            &\mathbf{v}_i(s+1) = \mathbf{v}_i(s) + \Delta t \cdot \mathbf{u}_i(s),\\
            &\underbrace{\|\mathbf{u}_i(s)\| \leq U_{\max}}_{\text{Acceleration limits}} , \quad \underbrace{\|\mathbf{v}_i(s)\| \leq V_{\max}}_{\text{Velocity limits}},\\
            &\underbrace{\|\mathbf{p}_i(s) - \mathbf{p}_j(s)\| \geq 2r_{\min} + \alpha \cdot (\tfrac{\|\mathbf{v}_i(s)\|^2}{2U_{\max}} + \tfrac{\|\mathbf{v}_j(s)\|^2}{2U_{\max}})}_{\text{Inter-drone separation}},\\
            &\underbrace{B_{r_i(s)}(\mathbf{p}_i(s)) \subset \Omega}_{\text{Airspace boundary}}.
        \end{aligned}
    \end{equation}
    The stage cost $\ell(\mathbf{x}_i, \mathbf{u}_i)$ penalizes goal deviations and control effort, and $\lambda > 0$ is a velocity penalty weight that incentivizes speed reduction, thereby shrinking safety zones in congested regions. Dynamics follow the discretized double integrator from Section~\ref{subsec:agent_model}. Collision safety requires inter-drone distances to exceed the sum of the adaptive safety radii from Definition~\ref{def:adaptive_radius}, and each safety zone must lie within $\Omega$.

    \subsection{Distributed MPC Formulation}\label{subsubsec:distributed_mpc}

    Centralized coordination is often infeasible: mixed fleets lack shared communication infrastructure, and privacy constraints may prevent full state disclosure. We therefore reformulate~\eqref{eq:centralized_mpc} as a distributed MPC, where each drone solves a local problem using only the estimated neighbor trajectories.

    At each decision time $k\Delta t$, drone $i$ solves its own MPC treating predicted neighbor trajectories $\hat{\mathbf{p}}_j(\cdot)$, $j \in \mathcal{N}_i$, as fixed:
    \begin{equation}\label{eq:distributed_mpc}
        \begin{aligned}
            \min_{\{\mathbf{u}_i(s)\}} \quad &\sum_{h=0}^{H-1} \left[\ell(\mathbf{x}_i(s), \mathbf{u}_i(s)) + \lambda\|\mathbf{v}_i(s)\|^2\right]\\
            \text{s.t.} \quad &\mathbf{p}_i(s+1) = \mathbf{p}_i(s) + \Delta t \cdot \mathbf{v}_i(s) + \tfrac{1}{2}\Delta t^2 \cdot \mathbf{u}_i(s),\\
            &\mathbf{v}_i(s+1) = \mathbf{v}_i(s) + \Delta t \cdot \mathbf{u}_i(s),\\
            &\underbrace{\|\mathbf{u}_i(s)\| \leq U_{\max}}_{\text{Acceleration limits}} , \quad \underbrace{\|\mathbf{v}_i(s)\| \leq V_{\max}}_{\text{Velocity limits}},\\
            &\underbrace{\|\mathbf{p}_i(s) - \hat{\mathbf{p}}_j(s)\| \geq r_i(\mathbf{v}_i(s)) + r_j(\hat{\mathbf{v}}_j(s)), \,\, \forall j \in \mathcal{N}_i}_{\text{Inter-drone separation with predicted neighbor trajectories}},\\
            &\underbrace{B_{r_i(s)}(\mathbf{p}_i(s)) \subset \Omega}_{\text{Airspace boundary}}.
        \end{aligned}
    \end{equation}

    Compared to~\eqref{eq:centralized_mpc}, the cost applies to a single agent, and the separation constraint replaces jointly optimized positions $\mathbf{p}_j$ with predicted trajectories $\hat{\mathbf{p}}_j$, where $r_i(\mathbf{v}_i) = r_{\min} + \alpha \frac{\|\mathbf{v}_i\|^2}{2U_{\max}}$ is the adaptive radius from Definition~\ref{def:adaptive_radius}. All other constraints are unchanged.

    \subsection{Iterative Solution Scheme in Distributed MPC}
    Since each drone solves its local MPC with neighbor trajectories fixed, the coupled problem has no closed-form solution. Instead, it is solved via asynchronous ADMM--Gauss--Seidel (ADMM-GS) iteration, combining ADMM decomposition with Gauss--Seidel sequential updates~\cite{stomberg2022dmpc_compendium, bertsekas1989parallel}. Let $\sigma^{(m)}$ denote a random permutation of $\{1, \ldots, N\}$ determining the processing order in iteration $m$. Each drone publishes its updated trajectory immediately upon completing its local solve; the next drone uses the most recent available trajectories:
    \begin{equation}\label{eq:gauss-seidel}
        \mathbf{u}_i^{(m+1)} = \arg\min_{\mathbf{u}_i} J_i\!\left(\mathbf{u}_i,\, \left\{\hat{\mathbf{p}}_j^{\text{latest}}\right\}_{j \in \mathcal{N}_i}\right),
    \end{equation}
    where $\hat{\mathbf{p}}_j^{\text{latest}}$ is the most recently published trajectory of neighbor $j$ — from iteration $m+1$ if $j$ has already been solved in the current round, otherwise from iteration $m$. Randomized ordering prevents systematic biases from a fixed sequence and serves as a symmetry-breaking mechanism. The iteration terminates when
    \begin{equation}
        \max_{i} \left\|\hat{\mathbf{p}}_i^{(m+1)} - \hat{\mathbf{p}}_i^{(m)}\right\| < \varepsilon
    \end{equation}
    for a tolerance $\varepsilon > 0$, or upon reaching a maximum iteration count.

    Having formalized the problem and both MPC formulations, we now establish feasibility and stability of the adaptive framework.
    The subsequent analysis relies equally on sphere packing theory and on MPC.
    Sphere packing theory provides the geometric capacity bounds that determine how many drones with adaptive safety radii can coexist within $\Omega$, while MPC enforces these bounds as trajectory constraints over a finite prediction horizon.
    Neither component alone is sufficient: the packing bounds define the feasibility limit, and MPC provides the control mechanism to operate safely up to that limit.


    \section{Stability and Feasibility Analysis}\label{sec:theory}
    We first derive the throughput-optimal crossing speed for narrow passages (Section~\ref{subsec:throughput}), which yields the critical drone density bound used to establish MPC feasibility (Section~\ref{subsec:feasibility}). Section~\ref{subsec:stability} then proves asymptotic stability of the centralized adaptive MPC via Lyapunov theory, and Section~\ref{subsec:distributed} extends both guarantees to the distributed setting.

    \subsection{Throughput Analysis for Narrow Passages}\label{subsec:throughput}

    Consider a planar cross-section $S \subset \Omega$ with area $|S|$ traversed at uniform speed~$v$.

    \begin{theorem}[Optimal Adaptive Throughput]\label{thm:adaptive_throughput}
        Let $r(v) = r_{\min} + \frac{\alpha v^2}{2U_{\max}}$ be the adaptive safety radius. In addition, let
        \begin{equation}\label{eq:geometric_capacity}
            n_S(v) = \left\lfloor \frac{|S|}{\pi\, r(v)^2} \right\rfloor
        \end{equation}
        be the maximum number of non-overlapping safety disks that fit on~$S$ (spatial capacity) and
        \begin{equation}\label{eq:temporal_separation}
            g(v) = \frac{2\,r(v)}{v}
        \end{equation}
        the minimum headway between successive drones in the same lane (temporal separation). Then the surface throughput $\Phi = n_S / g$ satisfies
        \begin{equation}\label{eq:throughput_adaptive}
            \Phi_{\text{a}}(v) = \frac{|S| \cdot v}{2\pi\left(r_{\min} + \frac{\alpha v^2}{2U_{\max}}\right)^3},
        \end{equation}
        which is maximized at the optimal crossing speed
        \begin{equation}\label{eq:v_opt_adaptive}
            v_{\text{a}}^* = \sqrt{\frac{2\,r_{\min}\,U_{\max}}{5\alpha}},
        \end{equation}
        provided $v_{\text{a}}^* \leq V_{\max}$. At this speed, the adaptive radius takes the value
        \begin{equation}\label{eq:r_opt}
            r(v_{\text{a}}^*) = \frac{6}{5}\,r_{\min},
        \end{equation}
        which is independent of $\alpha$, $U_{\max}$, and $V_{\max}$, and the maximum adaptive throughput is
        \begin{equation}\label{eq:phi_adaptive_max}
            \Phi_{\text{a}}^{\max} = \frac{125\,|S|}{432\,\pi\,r_{\min}^3} \cdot \sqrt{\frac{2\,r_{\min}\,U_{\max}}{5\alpha}}.
        \end{equation}
    \end{theorem}

    \begin{proof}
        Combining $\Phi = n_S / g$ and substituting the definitions yields~\eqref{eq:throughput_adaptive}. Differentiating with respect to~$v$ and setting to zero gives $r(v) - 3\alpha v^2/U_{\max} = 0$. Substituting $r(v) = r_{\min} + \frac{\alpha v^2}{2U_{\max}}$ yields $r_{\min} = \frac{5\alpha v^2}{2U_{\max}}$, hence~\eqref{eq:v_opt_adaptive}. Inserting $v_{\text{a}}^*$ into the radius formula gives $r(v_{\text{a}}^*) = r_{\min} + \frac{r_{\min}}{5} = \frac{6}{5}\,r_{\min}$, establishing~\eqref{eq:r_opt}. Substituting~\eqref{eq:v_opt_adaptive} and~\eqref{eq:r_opt} into~\eqref{eq:throughput_adaptive} yields~\eqref{eq:phi_adaptive_max}.
    \end{proof}

    \begin{corollary}[Static Throughput as Special Case]\label{cor:static_throughput}
        If the safety radius is fixed at $r_{\text{s}} = r(V_{\max}) = r_{\min} + \frac{\alpha V_{\max}^2}{2U_{\max}}$, the throughput reduces to
        \begin{equation}
            \Phi_{\text{s}}(v) = \frac{|S| \cdot v}{2\pi\, r_{\text{s}}^3},
        \end{equation}
        which is linear in~$v$ and attains its maximum at $v = V_{\max}$:
        \begin{equation}
            \Phi_{\text{s}}^{\max} = \frac{|S| \cdot V_{\max}}{2\pi\,r_{\text{s}}^3}.
        \end{equation}
    \end{corollary}

    \begin{theorem}[Adaptive Throughput Advantage]\label{thm:throughput}
        Let $S$ be a cross-section with area $|S|$ and let $r_{\text{s}} = r_{\min} + \frac{\alpha V_{\max}^2}{2U_{\max}}$.
        \begin{enumerate}
            \item \textbf{Exclusive passage regime:} If $\pi\,r_{\min}^2 \leq |S| < \pi\,r_{\text{s}}^2$, then $\Phi_{\text{s}} = 0$ since no drone with static radius can fit through~$S$. However, for any speed $v$ satisfying $r(v) \leq \sqrt{|S|/\pi}$, the adaptive throughput satisfies $\Phi_{\text{a}}(v) > 0$. In this regime, the throughput gain is unbounded.
            \item \textbf{Marginal passage regime:} If $|S| = \pi\,r_{\text{s}}^2(1+\varepsilon)$ for small $\varepsilon > 0$, then $n_S(r_{\text{s}}) = 1$ and the static throughput is limited by temporal separation alone:
            \begin{equation}
                \Phi_{\text{s}}^{\max} = \frac{V_{\max}}{2\,r_{\text{s}}}.
            \end{equation}
            With adaptive radii at the optimal speed $v_{\text{a}}^*$, the geometric capacity increases to $n_S(v_{\text{a}}^*) = \lfloor(1+\varepsilon)\,r_{\text{s}}^2/r(v_{\text{a}}^*)^2\rfloor$, yielding a throughput gain proportional to the ratio of cross-sectional areas.
        \end{enumerate}
    \end{theorem}

    \begin{proof}
        For part~(1), if $|S| < \pi\,r_{\text{s}}^2$, then $n_S(r_{\text{s}}) = 0$ and $\Phi_{\text{s}} = 0$ by~\eqref{eq:geometric_capacity}. Since $r(v) \to r_{\min}$ as $v \to 0$ and $|S| \geq \pi\,r_{\min}^2$ by assumption, there exists a speed $v_0 > 0$ such that $r(v_0) \leq \sqrt{|S|/\pi}$, giving $n_S(v_0) \geq 1$ and thus $\Phi_{\text{a}}(v_0) > 0$.

        For part~(2), with $|S| = \pi\,r_{\text{s}}^2(1+\varepsilon)$ and $n_S = 1$, the static throughput reduces to $1/g(V_{\max}) = V_{\max}/(2\,r_{\text{s}})$. For the adaptive case, at speed $v < V_{\max}$ we have $r(v) < r_{\text{s}}$ and hence
        \begin{equation}
            n_S(v) = \left\lfloor \frac{(1+\varepsilon)\,r_{\text{s}}^2}{r(v)^2} \right\rfloor \geq 1.
        \end{equation}
        As $v$ decreases, $r(v)$ shrinks and $n_S(v)$ grows, while $g(v) = 2r(v)/v$ changes according to~\eqref{eq:temporal_separation}. The optimal adaptive speed~\eqref{eq:v_opt_adaptive} balances geometric capacity against temporal separation.
    \end{proof}

    \begin{corollary}[Throughput Gain]    The throughput ratio $G = \Phi_{\text{a}}^{\max}/\Phi_{\text{s}}^{\max}$ satisfies:
        \begin{equation}
            G = \frac{r_{\text{s}}^3}{r(v_{\text{a}}^*)^3} \cdot \frac{v_{\text{a}}^*}{V_{\max}} = \frac{r_{\text{s}}^3}{\left(\frac{6}{5}\,r_{\min}\right)^3} \cdot \frac{v_{\text{a}}^*}{V_{\max}}.
        \end{equation}
        In the exclusive passage regime ($|S| < \pi\,r_{\text{s}}^2$), $G = \infty$.
    \end{corollary}

    \subsection{Feasibility of centralized MPC}\label{subsec:feasibility}

    We distinguish \emph{problem feasibility} (existence of a collision-free configuration) from \emph{solver feasibility} (ability of the MPC to find it within the prediction horizon).

    \subsubsection*{Critical Drone Density (Problem Feasibility)}

    The maximum drone count with uniform radius~$r$ in airspace~$\Omega$ follows from the sphere packing bound~\cite{drones10020139}
    \begin{equation}
                N_{\text{c}}(r) = \left\lfloor \delta_3 \cdot \frac{\mathrm{vol}(\Omega)}{\frac{4}{3}\pi r^3} \right\rfloor,
    \end{equation}
    where $\delta_3 = \frac{\pi}{\sqrt{18}} \approx 0.74$ is the optimal packing density in three dimensions. Since $r_{\min} \leq r_{\text{dyn}}(v) \leq r_{\text{s}}$ with $r_{\text{s}} = r_{\min} + \alpha V_{\max}^2/(2U_{\max})$, the drone count is bounded by
    \begin{equation}
                N_{\text{c}}(r_{\text{s}}) \leq N \leq N_{\text{c}}(r_{\min}).
    \end{equation}
    The upper bound $N_{\text{c}}(r_{\min})$ requires all drones to be stationary. Since drones must move to reach their goals, we define the critical adaptive drone density as
    \begin{equation}\label{eq:n_crit_var}
        N_{\text{c}}^{\text{v}} = N_{\text{c}}\!\left(r_{\min} + \varepsilon\right),
    \end{equation}
    where $\varepsilon > 0$ is arbitrarily small. In the limit $\varepsilon \to 0$, $N_{\text{c}}^{\text{v}} = N_{\text{c}}(r_{\min})$. Using the throughput-optimal radius $r(v_{\text{a}}^*) = \tfrac{6}{5}\,r_{\min}$ from~\eqref{eq:r_opt}, we define
    \begin{equation}\label{eq:n_opt}
        N_{\text{opt}} = N_{\text{c}}\!\left(\tfrac{6}{5}\,r_{\min}\right).
    \end{equation}
    For $N \geq N_{\text{opt}}$, throughput-optimal operation is infeasible; operating beyond $N_{\text{opt}}$ forces sub-optimal speeds and exponentially increases computational cost.

    \subsubsection*{Prediction Horizon (Solver Feasibility)}

    Even with $N \leq N_{\text{c}}^{\text{v}}$, the MPC requires a sufficiently long prediction horizon to plan beyond neighboring safety zones.

    \begin{lemma}[Minimum Horizon for Feasibility]\label{lem:horizon_feasibility}
        Assume $N \leq N_{\text{c}}^{\text{v}}$, so that a collision-free configuration exists. The MPC solver can find a feasible trajectory only if the prediction horizon satisfies
        \begin{equation}
            H \geq H_{\min} = \left\lceil \frac{1}{\Delta t}\sqrt{\frac{2r_{\min}\alpha}{U_{\max}}} \right\rceil.
        \end{equation}
    \end{lemma}

    \begin{proof}
        The solver can plan a collision-free maneuver only if the prediction window extends beyond the safety zone of each neighboring drone. This requires
        \begin{equation}
            H \cdot \Delta t \cdot \|\mathbf{v}_i\| > r_i(\mathbf{v}_i) = r_{\min} + \alpha \frac{\|\mathbf{v}_i\|^2}{2U_{\max}}, \quad \forall\, i,
        \end{equation}
        or equivalently,
        \begin{equation}
            H \cdot \Delta t > \frac{r_{\min}}{\|\mathbf{v}_i\|} + \frac{\alpha \|\mathbf{v}_i\|}{2U_{\max}}.
        \end{equation}
        To obtain a velocity-independent bound, we minimize the right-hand side $f(v) = r_{\min}/v + \alpha v/(2U_{\max})$:
        \begin{equation}
            f'(v) = -\frac{r_{\min}}{v^2} + \frac{\alpha}{2U_{\max}} = 0 \quad \Longrightarrow \quad v^* = \sqrt{\frac{2r_{\min}U_{\max}}{\alpha}},
        \end{equation}
        which yields $f(v^*) = \sqrt{2r_{\min}\alpha / U_{\max}}$. Hence the horizon must satisfy $H \cdot \Delta t > \sqrt{2r_{\min}\alpha / U_{\max}}$, and rounding up gives $H_{\min}$.
    \end{proof}

    In practice, multiple simultaneous conflicts consume parts of the horizon budget. The conflict resolution time between drones $i$ and $j$ is
    \begin{equation}
        \tau_{\text{r}}^{\text{a}}(\mathbf{v}_i, \mathbf{v}_j) = \frac{r_i(\mathbf{v}_i) + r_j(\mathbf{v}_j) + \Delta_{\text{sa}}}{v_d},
    \end{equation}
    where $v_d = \max(\|\mathbf{v}_i - \mathbf{v}_j\|, \epsilon_0)$ with $\epsilon_0 > 0$ is the relative speed between drones and $\Delta_{\text{sa}} = 0.1 \cdot r_{\min}$ is an additional safety margin. With adaptive zones, MPC feasibility depends both on the horizon length $H$ and on the velocity adjustment strategy. The maximum number of resolvable conflicts for drone $i$ is
    \begin{equation}
        \kappa_{\text{a}}(H, \mathcal{V}) = \left\lfloor\frac{H\Delta t}{\min_{j \in \mathcal{N}_i} \tau_{\text{r}}^{\text{a}}(\mathbf{v}_i, \mathbf{v}_j)}\right\rfloor,
    \end{equation}
    where $\mathcal{N}_i$ is the set of drones in potential conflict with drone~$i$. Since $\kappa_{\text{a}}$ grows linearly with $H$ but MPC cost grows at least quadratically, a trade-off between anticipation capacity and real-time performance arises.

    \subsection{Stability of centralized MPC}\label{subsec:stability}

    We prove convergence via Lyapunov's direct method~\cite{MAYNE2000789,BORRELLi}, constructing an energy function that accounts for the speed-dependent safety zones.

    \begin{definition}[Lyapunov Function for the Adaptive System]
    Let $\mathbf{x} = \{(\mathbf{p}_i, \mathbf{v}_i)\}_{i=1}^N$ be the total state of the system. The Lyapunov function is defined as
    \begin{equation}
        V(\mathbf{x}) = \sum_{i=1}^{N} \left[ \|\mathbf{p}_i - \bar{\mathbf{p}}_i\|^2 + \gamma \|\mathbf{v}_i\|^2 + \delta \left( r_i^2(t) - r_{\min}^2 \right) \right],
    \end{equation}
    where the target position of drone i is represented by $\bar{\mathbf{p}}_i$. The weighting parameter for the kinetic energy is denoted by $\gamma > 0$ and $\delta > 0$ represents the weighting parameter for the energy of the safety zone. The adaptive radius is given by $r_i(t) = r_i(\mathbf{v}_i(t)) = r_{\min} + \alpha \frac{\|\mathbf{v}_i\|^2}{2U_{\max}}$ (cf.\ Definition~\ref{def:adaptive_radius}). The shift $-r_{\min}^2$ ensures $V(\mathbf{x}^*) = 0$ at the equilibrium $\mathbf{x}^* = \{(\bar{\mathbf{p}}_i, \mathbf{0})\}_{i=1}^N$, where $r_i = r_{\min}$.
    \end{definition}

    Near equilibrium, MPC with no active constraints reduces to LQR~\cite{MAYNE2000789,BORRELLi}, enabling explicit computation of $\dot{V}$.

    \begin{theorem}[Stability of the Adaptive MPC]
        \label{thm:stability}
        The adaptive MPC system is asymptotically stable about the goal configuration $\{(\bar{\mathbf{p}}_i, 0)\}_{i=1}^N$, provided the following three conditions are jointly satisfied:
        \begin{enumerate}
            \item $\alpha < \alpha_{\text{c}} = \dfrac{2r_{\min}U_{\max}}{V_{\max}^2}$ \quad (adaptation parameter sufficiently small),
            \item $N < N_{\text{c}}^{\text{v}} = N_{\text{c}}\!\left(r_{\min} + \varepsilon\right)$ \quad (drone count below the adaptive packing limit),
            \item $H \geq H_{\min}$ \quad (prediction horizon sufficient for solver feasibility, cf.\ Lemma~\ref{lem:horizon_feasibility}).
        \end{enumerate}
    \end{theorem}
    Condition~1 constrains the velocity-to-radius coupling, Condition~2 constrains the geometric density, and Condition~3 ensures solver feasibility.

    \begin{proof}
        Conditions~2--3 ensure feasibility (Section~\ref{subsec:feasibility}). Given feasibility, we show $\dot{V} < 0$. The time derivative is:
        \begin{align}
            \dot{V} &= \sum_{i=1}^{N} \left[ \frac{d}{dt}\|\mathbf{p}_i - \bar{\mathbf{p}}_i\|^2 + \gamma \frac{d}{dt}\|\mathbf{v}_i\|^2 + \delta \frac{d}{dt}\left( r_i^2(t) - r_{\min}^2 \right) \right].
        \end{align}

        Computing each term with $\dot{\mathbf{p}}_i = \mathbf{v}_i$:
        \begin{align}
            \frac{d}{dt}\|\mathbf{p}_i - \bar{\mathbf{p}}_i\|^2 &= 2(\mathbf{p}_i - \bar{\mathbf{p}}_i)^T \dot{\mathbf{p}}_i = 2(\mathbf{p}_i - \bar{\mathbf{p}}_i)^T \mathbf{v}_i \\
            \frac{d}{dt}(r_i(t))^2 &= 2r_i(t) \frac{\partial r_i(t)}{\partial t}.
        \end{align}

        Differentiating $r_i(t) = r_{\min} + \alpha \frac{\|\mathbf{v}_i\|^2}{2U_{\max}}$:
        \begin{align}
            \frac{\partial r_i(t)}{\partial t}
            &= \frac{\alpha}{2U_{\max}} \cdot \frac{\partial \|\mathbf{v}_i\|^2}{\partial t}
            = \frac{\alpha}{2U_{\max}} \cdot 2\|\mathbf{v}_i\| \cdot \frac{\mathbf{v}_i^T \mathbf{u}_i}{\|\mathbf{v}_i\|} \\
            &= \frac{\alpha}{U_{\max}} \cdot \mathbf{v}_i^T \mathbf{u}_i.
        \end{align}
        Substituting and collecting terms:
        \begin{align}
            \dot{V} &= \sum_{i=1}^{N} \left[ 2(\mathbf{p}_i - \bar{\mathbf{p}}_i)^T \mathbf{v}_i + 2\gamma \mathbf{v}_i^T \mathbf{u}_i + 2\delta r_i(t) \cdot \frac{\alpha}{U_{\max}} \mathbf{v}_i^T \mathbf{u}_i \right] \\
            &= 2\sum_{i=1}^{N} \left[ (\mathbf{p}_i - \bar{\mathbf{p}}_i)^T \mathbf{v}_i + \left(\gamma + \frac{\delta \alpha r_i(t)}{U_{\max}}\right) \mathbf{v}_i^T \mathbf{u}_i \right].
        \end{align}

        Near equilibrium, safety radii shrink to $r_{\min}$ and separation constraints become inactive, so the MPC reduces to LQR~\cite{MAYNE2000789,BORRELLi}.

        \begin{lemma}[MPC Control Approximation]
            \label{lem:mpc_approx}
            It is a standard result in MPC theory~\cite{MAYNE2000789,BORRELLi} that, when no constraints are active, the finite-horizon MPC solution coincides with the infinite-horizon LQR feedback law. Applied to the adaptive safety sphere system, this means that near the goal configuration, where the pairwise separation constraints become inactive and the adaptive radii shrink to $r_{\min}$, the MPC selects controls that for sufficiently small $\epsilon > 0$ approximate the form
            \begin{equation}
                \mathbf{u}_i = -\mathbf{K}_p (\mathbf{p}_i - \bar{\mathbf{p}}_i) - \mathbf{K}_v \mathbf{v}_i + \mathcal{O}(\epsilon),
            \end{equation}
            where $\mathbf{K}_p, \mathbf{K}_v > 0$ are the effective gains that depend on the MPC cost matrices. The velocity penalty $\lambda\|\mathbf{v}_i\|^2$ in~\eqref{eq:centralized_mpc} increases the effective gain $\mathbf{K}_v$ via $\tilde{\mathbf{R}} = \mathbf{R} + \lambda \mathbf{I}$, which strengthens the damping and is consistent with the stability requirements derived below.
        \end{lemma}

        Substituting the MPC control with $\mathbf{u}_i \approx -\mathbf{K}_p(\mathbf{p}_i - \bar{\mathbf{p}}_i) - \mathbf{K}_v \mathbf{v}_i$:
        \begin{align}
            \dot{V} \leq 2\sum_{i=1}^{N} \Bigg[ &(\mathbf{p}_i - \bar{\mathbf{p}}_i)^T \mathbf{v}_i \\
            &- \left(\gamma + \frac{\delta \alpha r_i(t)}{U_{\max}}\right) \mathbf{v}_i^T (\mathbf{K}_p(\mathbf{p}_i - \bar{\mathbf{p}}_i) + \mathbf{K}_v \mathbf{v}_i) \Bigg].         \notag
        \end{align}

        Bounding the cross terms via Young's inequality:
        \begin{equation}
            (\mathbf{p}_i - \bar{\mathbf{p}}_i)^T \mathbf{v}_i \leq \frac{\|\mathbf{p}_i - \bar{\mathbf{p}}_i\|^2}{2\tau_p} + \frac{\tau_p \|\mathbf{v}_i\|^2}{2}
        \end{equation}
        for arbitrarily $\tau_p > 0$.

        And for the second cross term:
        \begin{equation}
            -\mathbf{v}_i^T \mathbf{K}_p (\mathbf{p}_i - \bar{\mathbf{p}}_i) \leq \frac{\mathbf{K}_p^2 \|\mathbf{p}_i - \bar{\mathbf{p}}_i\|^2}{2\tau_v} + \frac{\tau_v \|\mathbf{v}_i\|^2}{2}.
        \end{equation}

        After rearrangement:
        \begin{equation}
            \begin{split}
                \dot{V} \leq \sum_{i=1}^{N} \Bigg[
                    &-\left(\frac{1}{\tau_p} - \gamma_{\text{eff}} \frac{\mathbf{K}_p^2}{2\tau_v}\right) \|\mathbf{p}_i - \bar{\mathbf{p}}_i\|^2 \\
                    &- \left(\gamma_{\text{eff}} \mathbf{K}_v - \tau_p - \gamma_{\text{eff}}\frac{\tau_v}{2}\right) \|\mathbf{v}_i\|^2 \Bigg],
            \end{split}
        \end{equation}
        where $\gamma_{\text{eff}} = \gamma + \frac{\delta \alpha r_{\max}}{U_{\max}}$ with $r_{\max} = r_{\min} + \frac{\alpha V_{\max}^2}{2U_{\max}}$.

        Choosing $\tau_p = \frac{1}{2\mathbf{K}_p}$ and $\tau_v = \mathbf{K}_v$ to ensure $c_p, c_v > 0$:
        \begin{itemize}
            \item Condition $c_p > 0$: \quad $2\mathbf{K}_p > \gamma_{\text{eff}} \frac{\mathbf{K}_p^2}{2\mathbf{K}_v}$ \quad $\Leftrightarrow$ \quad $\mathbf{K}_v > \frac{\gamma_{\text{eff}} \mathbf{K}_p}{4}$
            \item Condition $c_v > 0$: \quad $\gamma_{\text{eff}} \mathbf{K}_v > \frac{1}{2\mathbf{K}_p} + \gamma_{\text{eff}}\frac{\mathbf{K}_v}{2}$ \quad $\Leftrightarrow$ \quad $\mathbf{K}_v > \frac{1}{\gamma_{\text{eff}} \mathbf{K}_p}$
        \end{itemize}

        Both hold simultaneously when $\gamma_{\text{eff}}$ is bounded, yielding:
        \begin{equation}
            \alpha < \frac{2r_{\min}U_{\max}}{V_{\max}^2} = \alpha_{\text{c}}.
        \end{equation}

        Under $\alpha < \alpha_{\text{c}}$, there exists $\lambda_{\min} > 0$ with:
        \begin{equation}
            \dot{V} \leq -\lambda_{\min} V,
        \end{equation}
        where $\lambda_{\min} = \min\!\left\{c_p,\; \dfrac{c_v}{\gamma + \delta r_{\min}\alpha / U_{\max}}\right\} > 0$. The denominator accounts for the velocity-dependent terms in $V$: near equilibrium, $r_i^2 - r_{\min}^2 \approx r_{\min}\alpha\|\mathbf{v}_i\|^2/U_{\max}$ to leading order, so $V \leq \sum_i [\|\mathbf{p}_i - \bar{\mathbf{p}}_i\|^2 + (\gamma + \delta r_{\min}\alpha/U_{\max})\|\mathbf{v}_i\|^2]$, and the decay bound follows. This guarantees exponential convergence to equilibrium.
    \end{proof}

    In summary, the above proof establishes asymptotic stability of the adaptive MPC, provided $\alpha < \alpha_{\text{c}}$. Far from equilibrium, the MPC is feasible because drones are well separated and safety constraints do not conflict; the system progresses toward the target under the constrained nonlinear control. Near the goal configuration, the MPC reduces to LQR, which enables the explicit Lyapunov argument.

    \subsection{Extension to Distributed Optimization}\label{subsec:distributed}

    We extend the centralized results to the distributed MPC of Section~\ref{subsubsec:distributed_mpc}, where each drone solves a local problem using neighbor trajectories via asynchronous Gauss--Seidel updates.

    \subsubsection*{Feasibility of distributed MPC}

    The bounds $N_{\text{c}}^{\text{v}}$~\eqref{eq:n_crit_var} and $H_{\min}$ (Lemma~\ref{lem:horizon_feasibility}) are per-drone conditions that carry over directly. Convergence to a feasible fixed point is guaranteed by the contraction property (Lemma~\ref{lem:contraction}).

    \subsubsection*{Stability of distributed MPC}

    For each drone $i$, define the local update mapping:
    \begin{equation}
        \mathbf{u}_i^* = \mathcal{T}_i\!\left(\left\{\hat{\mathbf{p}}_j\right\}_{j \in \mathcal{N}_i}\right).
    \end{equation}
    A fixed point $\mathbf{u}^*$ of the full mapping $\mathcal{T} = (\mathcal{T}_1, \ldots, \mathcal{T}_N)$ corresponds to locally optimal trajectories, matching the iteration in~\eqref{eq:gauss-seidel}.

    \begin{lemma}[Contraction of the Local Update Map]
        \label{lem:contraction}
        Suppose that for each drone $i$, at most $\kappa$ collision avoidance constraints are simultaneously active. Then the best-response mapping $\mathcal{T}$ is a contraction in the weighted maximum norm:
        \begin{equation}
            \|\mathcal{T}(\mathbf{u}) - \mathcal{T}(\mathbf{u}')\|_\infty^w \leq \beta \|\mathbf{u} - \mathbf{u}'\|_\infty^w, \quad \beta < 1,
        \end{equation}
        with the explicit upper bound
        \begin{equation}\label{eq:beta_bound}
            \beta \leq \frac{\kappa \cdot \sup_{j \in \mathcal{N}_i} \mu_j}{\sigma_{\min}\!\left(\nabla^2_{\mathbf{u}_i} \mathcal{L}_i\right)},
        \end{equation}
        where $\mu_j \geq 0$ is the Lagrange multiplier of the collision constraint for pair $(i,j)$ and $\sigma_{\min}$ denotes the smallest singular value of the Hessian of the Lagrangian with respect to the control variables. In particular, $\beta \to 0$ as the inter-drone distances increase, since the multipliers $\mu_j$ vanish for inactive constraints.
    \end{lemma}

    \begin{proof}
        The local cost $J_i$ depends on the neighbor trajectories only through the collision avoidance constraints. By the standard sensitivity analysis for parametric nonlinear programs~\cite{bonnans2000perturbation}, applying the implicit function theorem to the KKT conditions yields a per-neighbor Lipschitz bound $\mu_j / \sigma_{\min}(\nabla^2_{\mathbf{u}_i} \mathcal{L}_i)$. Summing over at most $\kappa$ active neighbors gives~\eqref{eq:beta_bound}. The condition $\beta < 1$ holds whenever the MPC cost matrices $\mathbf{Q}$ and $\mathbf{R}$ provide sufficient regularization relative to the constraint coupling. The Banach fixed-point theorem then guarantees existence and uniqueness of the fixed point $\mathbf{u}^*$. For the asynchronous Gauss--Seidel iteration, convergence to the same fixed point follows from the classical theory of asynchronous iterations~\cite{bertsekas1989parallel}, since the scheme updates each drone exactly once per round with delay at most $N$.
    \end{proof}

    The following Lemma quantifies how distributed coupling perturbs the LQR approximation of Lemma~\ref{lem:mpc_approx}.

    \begin{lemma}[Perturbed LQR Approximation under Distributed Optimization]
        \label{lem:perturbed_lqr}
        At the fixed point $\mathbf{u}^*$ of the distributed iteration, each local MPC solution approximates the LQR feedback law with a coupling perturbation:
        \begin{equation}
            \mathbf{u}_i^* = -\mathbf{K}_p(\mathbf{p}_i - \bar{\mathbf{p}}_i) - \mathbf{K}_v \mathbf{v}_i + \eta_i,
        \end{equation}
        where the perturbation satisfies $\|\eta_i\| \leq C \cdot \beta$ with $C > 0$ depending on the constraint geometry.
    \end{lemma}

    \begin{proof}
        By Lemma~\ref{lem:mpc_approx}, the unconstrained local MPC solution has the linear feedback form $\mathbf{u}_i = -\mathbf{K}_p(\mathbf{p}_i - \bar{\mathbf{p}}_i) - \mathbf{K}_v \mathbf{v}_i$, since each drone's local cost function is identical to the per-agent component of the centralized cost. The perturbation $\eta_i$ arises from the coupling through the collision avoidance constraints: at the fixed point, the KKT conditions of the local problem include complementarity conditions involving the neighbor trajectories. By Lemma~\ref{lem:contraction}, the sensitivity of these conditions to neighbor perturbations is bounded by $\beta$, yielding the stated bound on $\|\eta_i\|$.
    \end{proof}

    \begin{theorem}[Stability of the Distributed Adaptive MPC]
        \label{thm:distributed_stability}
        Under the conditions of Theorem~\ref{thm:stability} ($\alpha < \alpha_{\text{c}}$, $N < N_{\text{c}}^{\text{v}}$, $H \geq H_{\min}$) and the additional coupling condition
        \begin{equation}
                    \beta < \beta_{\text{c}} = \frac{c_v}{\gamma_{\text{eff}} \cdot C},
        \end{equation}
        where $c_v$ is the velocity dissipation coefficient from Theorem~\ref{thm:stability}, $\gamma_{\text{eff}} = \gamma + \frac{\delta \alpha r_{\max}}{U_{\max}}$, and $C > 0$ is the constraint geometry constant from Lemma~\ref{lem:perturbed_lqr}, the distributed MPC system with asynchronous Gauss--Seidel updates is asymptotically stable about the goal configuration $\{(\bar{\mathbf{p}}_i, 0)\}_{i=1}^N$.
    \end{theorem}

    \begin{proof}
        We use the same Lyapunov function $V(\mathbf{x})$ from Theorem~\ref{thm:stability}. Substituting the perturbed control law from Lemma~\ref{lem:perturbed_lqr} into the time derivative yields
        \begin{equation}
            \dot{V}_{\text{dist}} = \dot{V}_{\text{central}} + 2\sum_{i=1}^{N} \gamma_{\text{eff}} \, \mathbf{v}_i^T \eta_i,
        \end{equation}
        where $\dot{V}_{\text{central}}$ is the derivative under the centralized control from Theorem~\ref{thm:stability} and the second term is the coupling error from the distributed optimization.

        Applying the Cauchy--Schwarz inequality:
        \begin{equation}
            \left|\sum_{i=1}^{N} \gamma_{\text{eff}} \, \mathbf{v}_i^T \eta_i\right| \leq \gamma_{\text{eff}} \, C \, \beta \sum_{i=1}^{N} \|\mathbf{v}_i\|.
        \end{equation}
        From the centralized analysis, $\dot{V}_{\text{central}}$ contains a negative-definite term $-c_v \sum_i \|\mathbf{v}_i\|^2$ arising from the velocity damping through $\mathbf{K}_v$. Combining both contributions:
        \begin{align}
            \dot{V}_{\text{dist}} &\leq -c_v \sum_{i=1}^{N} \|\mathbf{v}_i\|^2 + \gamma_{\text{eff}} \, C \, \beta \sum_{i=1}^{N} \|\mathbf{v}_i\| \nonumber\\
            &= \sum_{i=1}^{N} \|\mathbf{v}_i\|\left(-c_v \|\mathbf{v}_i\| + \gamma_{\text{eff}} \, C \, \beta\right).
        \end{align}
        Each summand is negative when $\|\mathbf{v}_i\| > \gamma_{\text{eff}} C \beta / c_v$. For the regime near equilibrium in which velocities are small, the position damping $-c_p \sum_i \|\mathbf{p}_i - \bar{\mathbf{p}}_i\|^2$ from the centralized analysis dominates. Therefore $\dot{V}_{\text{dist}} < 0$ outside a neighborhood of the equilibrium whose radius is proportional to $\beta$. Since $\beta \to 0$ as the inter-drone distances grow, asymptotic stability holds when
        \begin{equation}
            \beta < \beta_{\text{c}} = \frac{c_v}{\gamma_{\text{eff}} \cdot C}.
        \end{equation}
        This argument follows the standard framework for stability of perturbed MPC systems~\cite{MAYNE2000789}, combined with the contraction-based analysis of distributed optimization schemes~\cite{ROSTAMI2023100881}.
    \end{proof}

    The bound~\eqref{eq:beta_bound} is verifiable from the MPC cost matrices and KKT solution. In practice, $\beta$ is small when drone density is low ($\kappa$ small), constraints are weakly active ($\mu_j$ small), or cost regularization is strong ($\sigma_{\min}$ large). Hence $N_{\text{c}}^{\text{v}}$~\eqref{eq:n_crit_var} serves as a sufficient condition for both centralized stability and distributed convergence.


    \section{Evaluation}\label{sec:evaluation}

    Simulation \footnote{Source code and configurations are available at \url{https://github.com/linda78/TLCAmpc/releases/tag/v2.2.1}.} bridges the gap between the idealized assumptions of Section~\ref{sec:theory} and practical applicability by exposing effects such as discrete time steps, finite solver iterations, and transient radius expansion during acceleration.

    We validate the theoretical predictions through two experiments. The first (Section~\ref{subsec:eval_throughput}) targets the throughput predictions by placing drones in constrained passage scenarios. The second (Section~\ref{subsec:eval_scaling}) addresses the capacity and computational scaling claims by varying the number of drones. All experiments share the physical parameters summarized in Table~\ref{tab:parameters}. The static radius equals the worst-case adaptive radius at maximum velocity: $r_{\text{s}} = r_{\min} + \alpha V_{\max}^2/(2U_{\max})$. The velocity and acceleration limits match~\cite{drones10020139}, and $\alpha = 0.5$ lies below the critical threshold $\alpha_{\text{c}} = 0.96$ (Theorem~\ref{thm:stability}).

    Table~\ref{tab:parameters} also lists three capacity bounds for the scaling experiment. The static bound $N_{\text{c}}^{\text{s}} = 10$ uses the dynamically inflated radius from~\cite{drones10020139}. The throughput-optimal operating point $N_{\text{opt}} = 18$ is evaluated at $r(v_{\text{a}}^*) = \tfrac{6}{5}\,r_{\min}$ from~\eqref{eq:n_opt}, maximizing the product of drone count and velocity. Below this bound, the MPC solver converges quickly; beyond it, drones must operate below the throughput-optimal speed, and computation times grow rapidly. The stability limit $N_{\text{c}}^{\text{v}} = 25$ corresponds to the adaptive packing bound~\eqref{eq:n_crit_var} at $r = r_{\min} + \varepsilon$. Below this limit, asymptotic stability is guaranteed (Theorem~\ref{thm:stability}), but exponential computation time growth restricts the practically achievable drone count well below this theoretical bound.

    \begin{table}[ht!]
        \centering
        \caption{Simulation parameters and derived capacity bounds shared by both experiments.}
        \small
        \renewcommand{\arraystretch}{1.15}
        \begin{tabular}{l l l}
            \hline
            \textbf{Parameter}            & \textbf{Symbol}           & \textbf{Value}        \\
            \hline
            Maximum velocity              & $V_{\max}$                & $2.5\,\mathrm{m/s}$   \\
            Maximum acceleration          & $U_{\max}$                & $3.0\,\mathrm{m/s}^2$ \\
            Minimum safety radius         & $r_{\min}$                & $1.0\,\mathrm{m}$     \\
            Static safety radius          & $r_{\text{s}}$            & $1.52\,\mathrm{m}$    \\
            Adaptation parameter          & $\alpha$                  & $0.5$                 \\
            \hline
            \multicolumn{3}{l}{\textit{Derived capacity bounds (Experiment~2)}} \\
            Static (dynamic packing)      & $N_{\text{c}}^{\text{s}}$ & $10$                  \\
            Adaptive (throughput-optimal) & $N_{\text{opt}}$          & $18$                  \\
            Adaptive (stability limit)    & $N_{\text{c}}^{\text{v}}$ & $25$                  \\
            \hline
        \end{tabular}
        \label{tab:parameters}
    \end{table}

    Both experiments compare four configurations combining two safety radius strategies and two MPC architectures in a $2 \times 2$ factorial design (Table~\ref{tab:configurations}). The static radius equals the maximum-velocity adaptive radius, ensuring the adaptive variant can only be less restrictive.

    \begin{table}[ht!]
        \centering
        \caption{Factorial design: safety radius strategy $\times$ MPC architecture.}
        \small
        \renewcommand{\arraystretch}{1.15}
        \begin{tabular}{c l l}
            \hline
            \textbf{Config.} & \textbf{Safety Radius}   & \textbf{MPC Architecture}   \\
            \hline
            A                & adaptive $r(\mathbf{v})$ & centralized                 \\
            B                & adaptive $r(\mathbf{v})$ & distributed (Gauss--Seidel) \\
            C                & static $r_{\text{s}}$    & centralized                 \\
            D                & static $r_{\text{s}}$    & distributed (Gauss--Seidel) \\
            \hline
        \end{tabular}
        \label{tab:configurations}
    \end{table}


    \subsection{Throughput in Constrained Passages}\label{subsec:eval_throughput}

    This experiment validates the throughput analysis of Section~\ref{subsec:throughput} using two passage geometries that correspond to the exclusive and marginal regimes of Theorem~\ref{thm:throughput}.

    \subsubsection*{Numerical Predictions}

    With the parameters from Table~\ref{tab:parameters}, the optimal adaptive crossing speed is $v_{\text{a}}^{*} \approx 1.55\,\text{m/s}$ at radius $r(v_{\text{a}}^{*}) = 1.2\,\text{m}$. A passage with $|S| = 5.0\,\text{m}^2$ falls in the exclusive regime: it is impassable with static spheres ($\pi\,r_{\text{s}}^2 \approx 7.26\,\text{m}^2 > |S|$) but navigable with adaptive spheres at reduced speed. In the marginal regime, the throughput gain for a wider passage ($|S| = 8.0\,\text{m}^2$) evaluates to $G \approx 1.26$, i.e., approximately $26\,\%$ higher throughput with adaptive spheres.

    \subsubsection*{Setup}

    \paragraph{Exclusive passage}
    A wall with a single square hole opening of side length $\ell = 2.4\,\mathrm{m}$ is placed in the airspace. Since $2r_{\min} = 2.0\,\mathrm{m} \leq \ell$ but $2r_{\text{s}} = 3.04\,\mathrm{m} > \ell$,
    the opening lies in the exclusive regime: static drones ($r_{\text{s}} = 1.52\,\mathrm{m}$) cannot fit through, while adaptive drones ($r_{\min} = 1.0\,\mathrm{m}$) can pass by reducing speed.

    \paragraph{Narrow tunnel}
    Five drones traverse a square tunnel (side $3.65\,\mathrm{m}$). One static drone fits ($2 r_{\text{s}} = 3.04\,\mathrm{m} < 3.65\,\mathrm{m}$) but two do not ($4 r_{\text{s}} = 6.08\,\mathrm{m} > 3.65\,\mathrm{m}$), forcing sequential passage. Adaptive radii allow tighter spacing when drones slow down.

    \subsubsection*{Results}

    \paragraph{Exclusive passage}
    \begin{figure}[ht!]
        \centering
        \begin{minipage}[t]{0.48\linewidth}
            \centering
            \includegraphics[width=\linewidth]{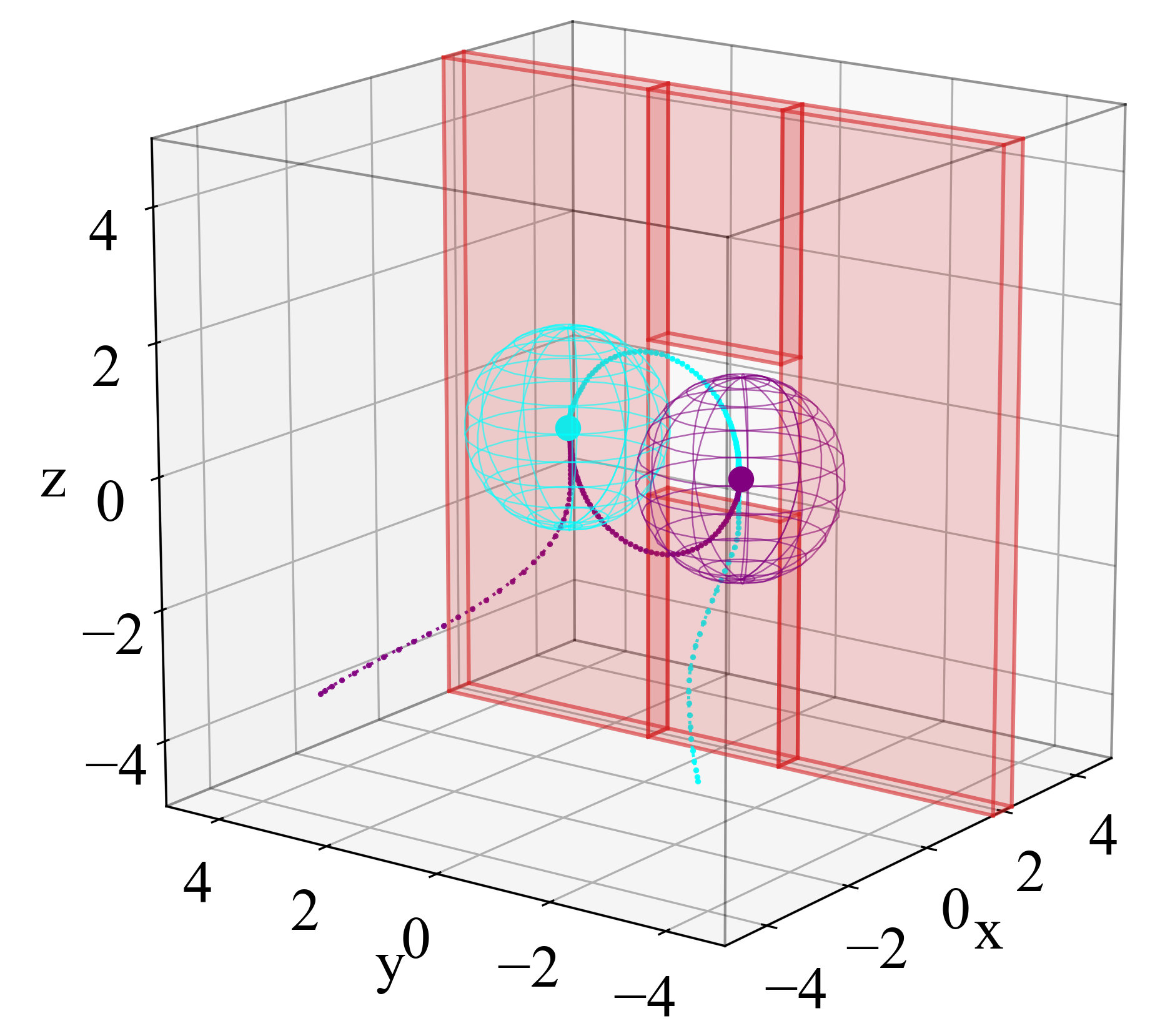}
            \caption{Static safety sphere ($r_{\text{s}} = 1.52\,\mathrm{m}$) at a passage with radius $1.2\,\mathrm{m}$. The drone cannot fit through the opening.}
            \label{fig:static_exclusive}
        \end{minipage}
        \hfill
        \begin{minipage}[t]{0.48\linewidth}
            \centering
            \includegraphics[width=\linewidth]{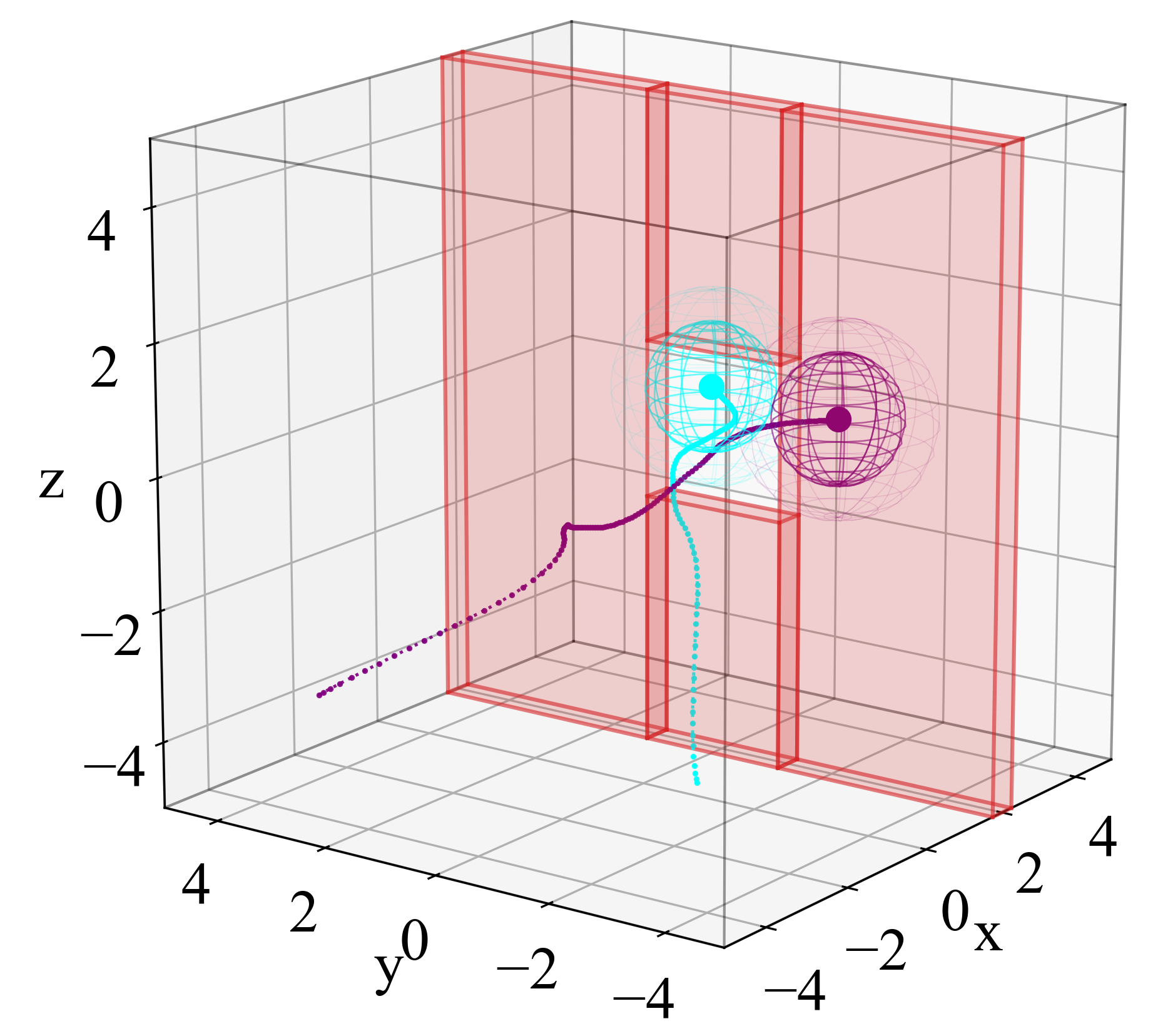}
            \caption{Adaptive safety sphere ($r_{\min} = 1.0\,\mathrm{m}$) at the same passage. The drone reduces its speed, shrinks its safety radius, and traverses the $1.2\,m$ opening.}
            \label{fig:adaptive_exclusive}
        \end{minipage}
    \end{figure}
    Figure~\ref{fig:static_exclusive} confirms that a drone with a static safety sphere cannot traverse the passage. In contrast, Figure~\ref{fig:adaptive_exclusive} shows that the adaptive drone reduces its speed and shrinks its safety radius below the opening radius, successfully passing through. In this scenario, the DMPC with adaptive spheres (\cfg{B}) outperformed the centralized variant (\cfg{A}), suggesting that local optimization effectively exploits adaptive radius reduction near obstacles.

    \paragraph{Narrow tunnel}
    \begin{figure}[ht!]
        \centering
        \begin{minipage}[t]{0.48\linewidth}
            \centering
            \includegraphics[width=\linewidth]{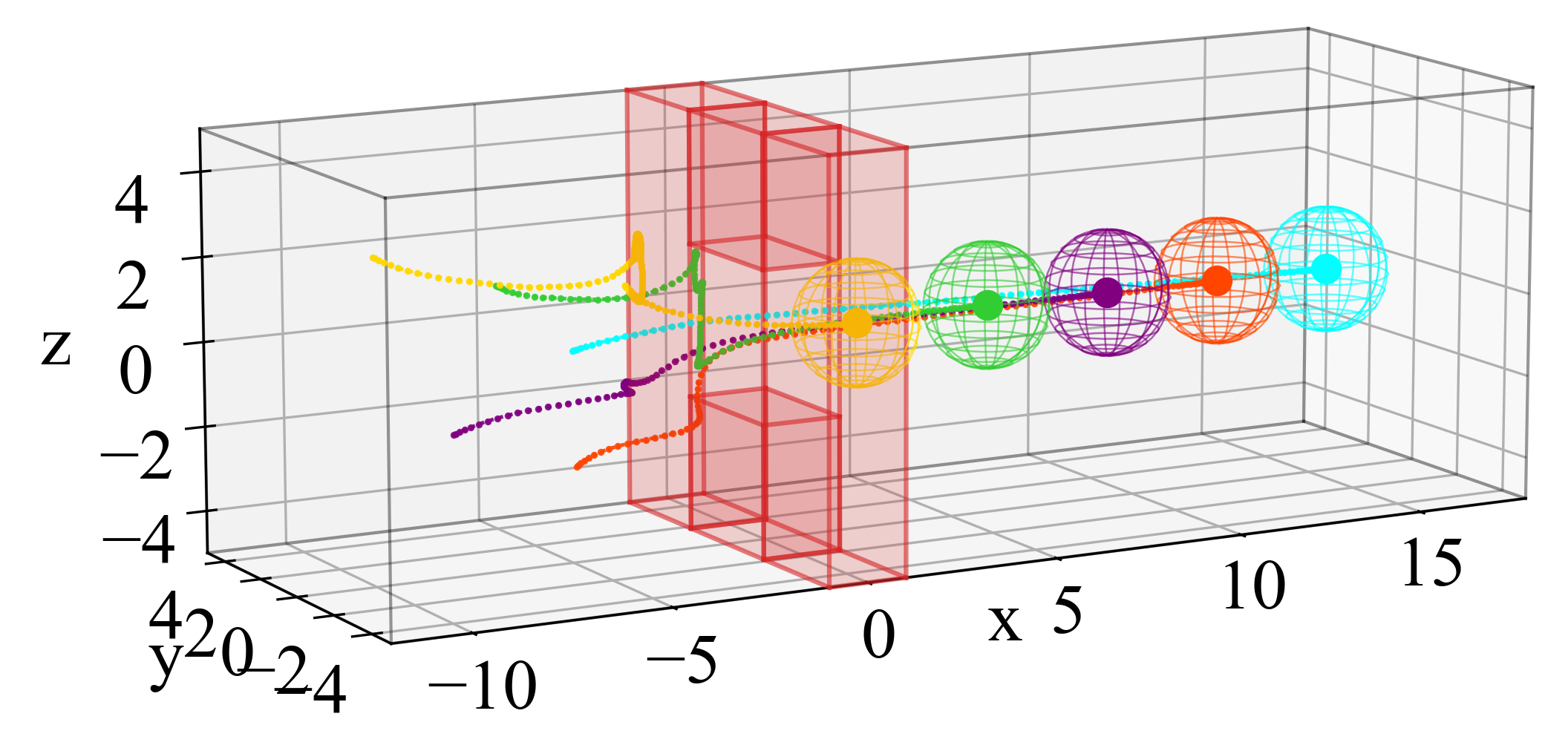}
            \caption{Five drones traversing a square tunnel (side length $3.65\,\mathrm{m}$) with static safety radii. Total traversal time: ${\approx}\,12\,\mathrm{s}$.}
            \label{fig:static_narrow}
        \end{minipage}
        \hfill
        \begin{minipage}[t]{0.48\linewidth}
            \centering
            \includegraphics[width=\linewidth]{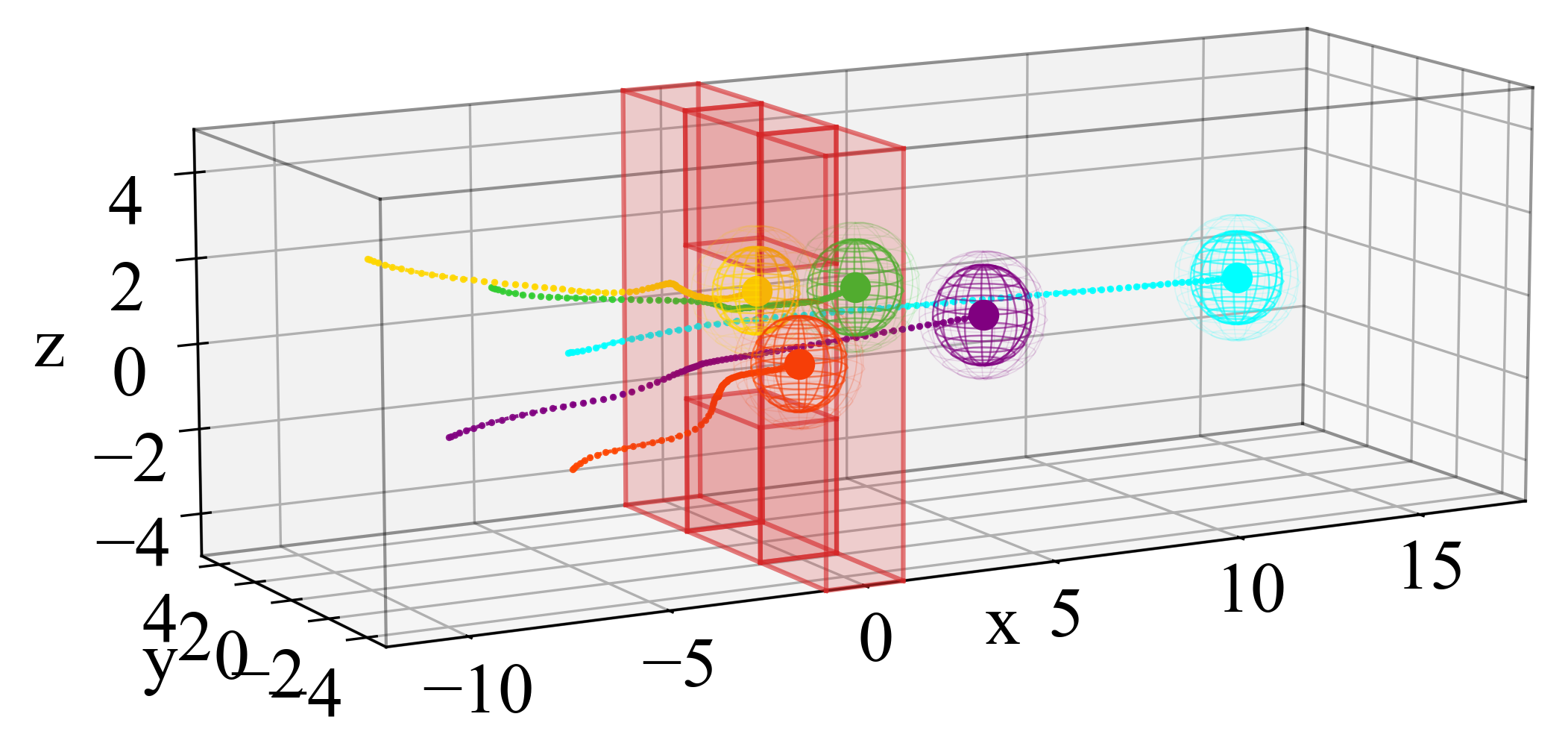}
            \caption{Same scenario with adaptive safety zones. The reduced radii allow tighter spacing and faster throughput. Total traversal time: ${\approx}\,9\,\mathrm{s}$.}
            \label{fig:adaptive_narrow}
        \end{minipage}
    \end{figure}
    With static safety radii, all five drones require approximately $12$~seconds to traverse (Figure~\ref{fig:static_narrow}). With adaptive safety zones, the drones slow down when approaching the tunnel, shrinking their safety radii and allowing tighter spacing. This reduces the total traversal time to about nine seconds (Figure~\ref{fig:adaptive_narrow}), a reduction of $25\,\%$.

    \paragraph{Summary.}
    Both scenarios confirm Theorem~\ref{thm:throughput}: adaptive spheres enable passage through openings impassable with static radii (exclusive regime) and yield measurable throughput gains in constrained spaces (marginal regime).


    \subsection{Density Scaling and Computational Cost}\label{subsec:eval_scaling}

    This experiment evaluates capacity bounds and computational scaling as drone count increases, addressing two questions:
    \begin{enumerate}
        \item[\textbf{Q1}] At what drone count does feasibility break down for static versus adaptive configurations?
        \item[\textbf{Q2}] How does computational cost scale with drone count, and how do centralized and DMPC compare?
    \end{enumerate}

    \subsubsection*{Setup}

    All runs use a cubic airspace $\Omega = [0,8]^3\,\mathrm{m}^3$ with randomly selected start and goal positions satisfying~\eqref{eq:pairwise_separation}. Identical random seeds ensure the same initial conditions across configurations. We vary $N \in \{2, \ldots, N_{\text{c}}^{\text{v}}\}$, repeating each combination $20$ times ($4 \times 25 \times 20 = 2\,000$ runs).

    We record mean step time (wall-clock per MPC solve), total steps to completion, and feasibility rate (fraction of converging runs). For distributed configurations~\cfg{(B, D)}, the number of Gauss--Seidel iterations per step is also recorded.

    At the first time step, each drone is initialized with a constant trajectory $\hat{\mathbf{p}}_i^{(0)}(k+h) = \mathbf{p}_i(k)$ plus small symmetry-breaking perturbations; subsequent steps use warm-started solutions shifted forward by one step.

    \subsubsection*{Results}

    \paragraph{Feasibility and Deadlock Rate (Q1)}

    \begin{figure}[ht!]
        \centering
        \includegraphics[width=0.9\linewidth]{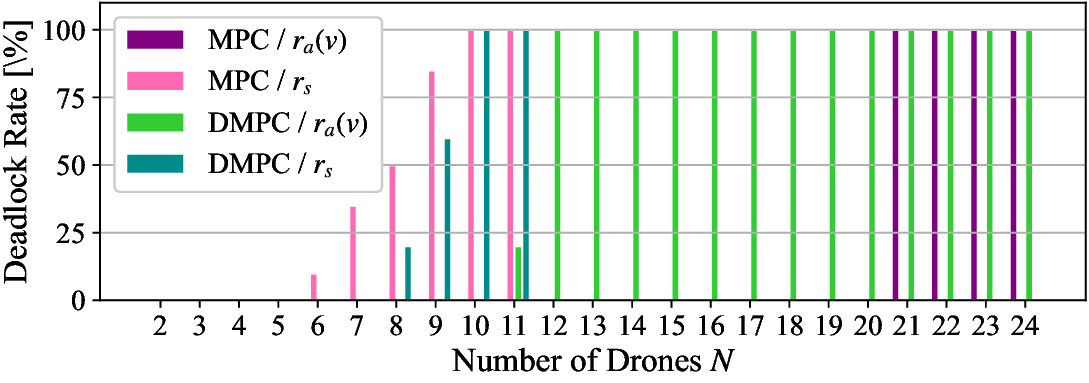}
        \caption{Deadlock rate as a function of drone count~$N$ for all four configurations. Static configurations fail progressively and reach $100\,\%$ deadlock at $N = 10$; adaptive configurations extend the feasible range, with~\cfg{A} operating up to $N = 20$. Static spheres cannot be placed without overlap for $N \geq 11$.}
        \label{fig:deadlock_rate}
    \end{figure}

    Figure~\ref{fig:deadlock_rate} shows the fraction of scenarios that terminate in deadlock or timeout for each configuration as a function of~$N$. The four configurations separate clearly.

    The centralized static variant~\cfg{C} fails first, with deadlocks from $N = 6$. The distributed static variant~\cfg{D} follows at $N = 8$. At $N = 10$, both static configurations reach $100\,\%$ failure, confirming the theoretical limit $N_{\text{c}}^{\text{s}} = 10$.

    The adaptive configurations extend the feasible range substantially. The centralized adaptive variant~\cfg{A} resolves all scenarios up to $N = 20$, beyond which every run times out. This lies between $N_{\text{opt}} = 18$ and $N_{\text{c}}^{\text{v}} = 25$, consistent with theory. The bound $N_{\text{opt}}$ assumes all drones operate at the throughput-maximizing speed $v_{\text{a}}^*$. In practice, the centralized solver exploits heterogeneous velocity profiles — slowing individual drones to resolve local conflicts while others continue at higher speeds — accommodating more drones than $N_{\text{opt}}$ at the cost of reduced per-drone throughput and increasing computation time. Feasibility breaks down at $N = 21$ rather than $N_{\text{c}}^{\text{v}} = 25$ due to computational limitations: exponential solver time growth (Figure~\ref{fig:scaling_mpc}) causes runs to exceed the timeout before convergence, not stability failure.

    The distributed adaptive variant~\cfg{B} fails earlier: at $N = 11$, approximately $20\,\%$ of runs exceed the step limit or time out, and no scenario with $N \geq 12$ is solved. This earlier breakdown reflects the coordination overhead of Gauss--Seidel iteration: without global state access, iterative trajectory updates require more rounds to converge, exhausting the computation budget at lower drone counts.

    These results confirm three predictions. First, adaptive spheres extend the feasible range by a factor of~$2$ over static radii, as predicted by the capacity analysis. Second, the achievable drone count exceeds $N_{\text{opt}} = 18$ but remains below $N_{\text{c}}^{\text{v}} = 25$, confirming that $N_{\text{opt}}$ is a conservative but practically meaningful bound while $N_{\text{c}}^{\text{v}}$ correctly delimits the stability regime. Third, although centralized and DMPC share identical theoretical guarantees (Theorems~\ref{thm:stability} and~\ref{thm:distributed_stability}), the centralized architecture realizes a substantially larger fraction of the theoretical capacity due to more efficient use of global state information.

    \paragraph{Computational Cost (Q2)}

    \begin{figure}[ht!]
        \centering
        \includegraphics[width=0.9\linewidth]{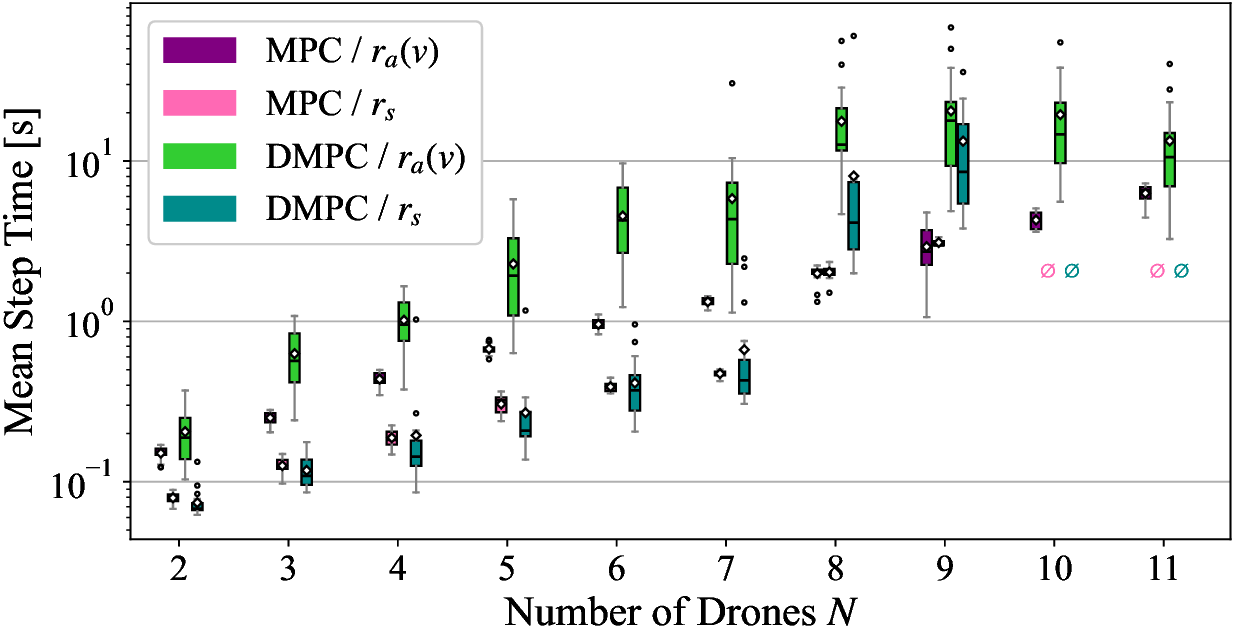}
        \caption{Mean step time (logarithmic scale) as a function of drone count~$N$ for all four configurations. All configurations exhibit exponential growth; adaptive safety zones and the distributed architecture incur the highest computational cost.}
        \label{fig:mean_step_time_all}
    \end{figure}

    Figure~\ref{fig:mean_step_time_all} shows the mean step time on a logarithmic scale as a function of~$N$. All four configurations transition to an approximately linear trend in the log domain beyond $N \approx 5$, i.e., once more than half of the static capacity is utilized. This confirms that the mean computation time per step grows exponentially with the number of drones.

    The distributed adaptive variant~\cfg{B} stands out with the steepest slope, indicating the worst scaling behavior. The combination of iterative Gauss--Seidel coordination and nonlinear, velocity-dependent separation constraints compounds the per-step cost, making this configuration the most computationally demanding. The centralized architectures~\cfg{(A, C)} scale more favorably, resolving all pairwise constraints simultaneously with full state access.

    A second notable observation is the substantially larger variance of the adaptive configurations~\cfg{(A, B)} compared to the static ones~\cfg{(C, D)}. With static safety radii, computation time is largely determined by drone count and varies little across initial placements. With adaptive radii, the initial constellation strongly influences computation: unfavorable start-goal assignments force prolonged deceleration and repeated conflict resolution, producing step times significantly above favorable cases.

    Exponential computation time growth is the primary practical limitation of both architectures. Figures~\ref{fig:scaling_mpc} and~\ref{fig:scaling_dmpc} isolate the scaling of centralized and distributed adaptive variants.

    \begin{figure}[ht!]
        \centering
        \begin{minipage}[t]{0.48\linewidth}
            \centering
            \includegraphics[width=\linewidth]{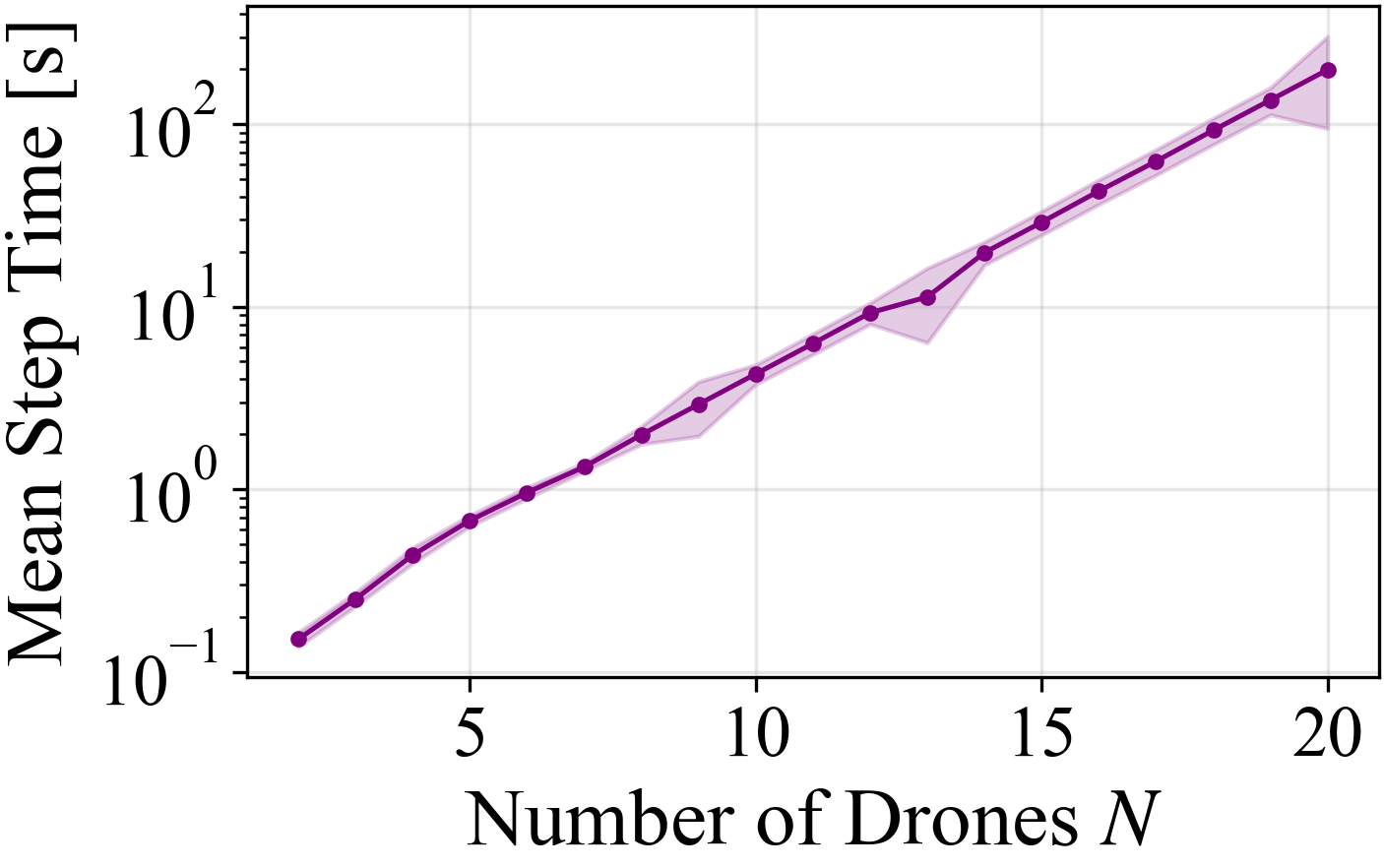}
            \caption{Computation time per step (logarithmic scale) for centralized adaptive~\cfg{A} over $N \in \{2,\ldots,20\}$.
            }
            \label{fig:scaling_mpc}
        \end{minipage}
        \hfill
        \begin{minipage}[t]{0.48\linewidth}
            \centering
            \includegraphics[width=\linewidth]{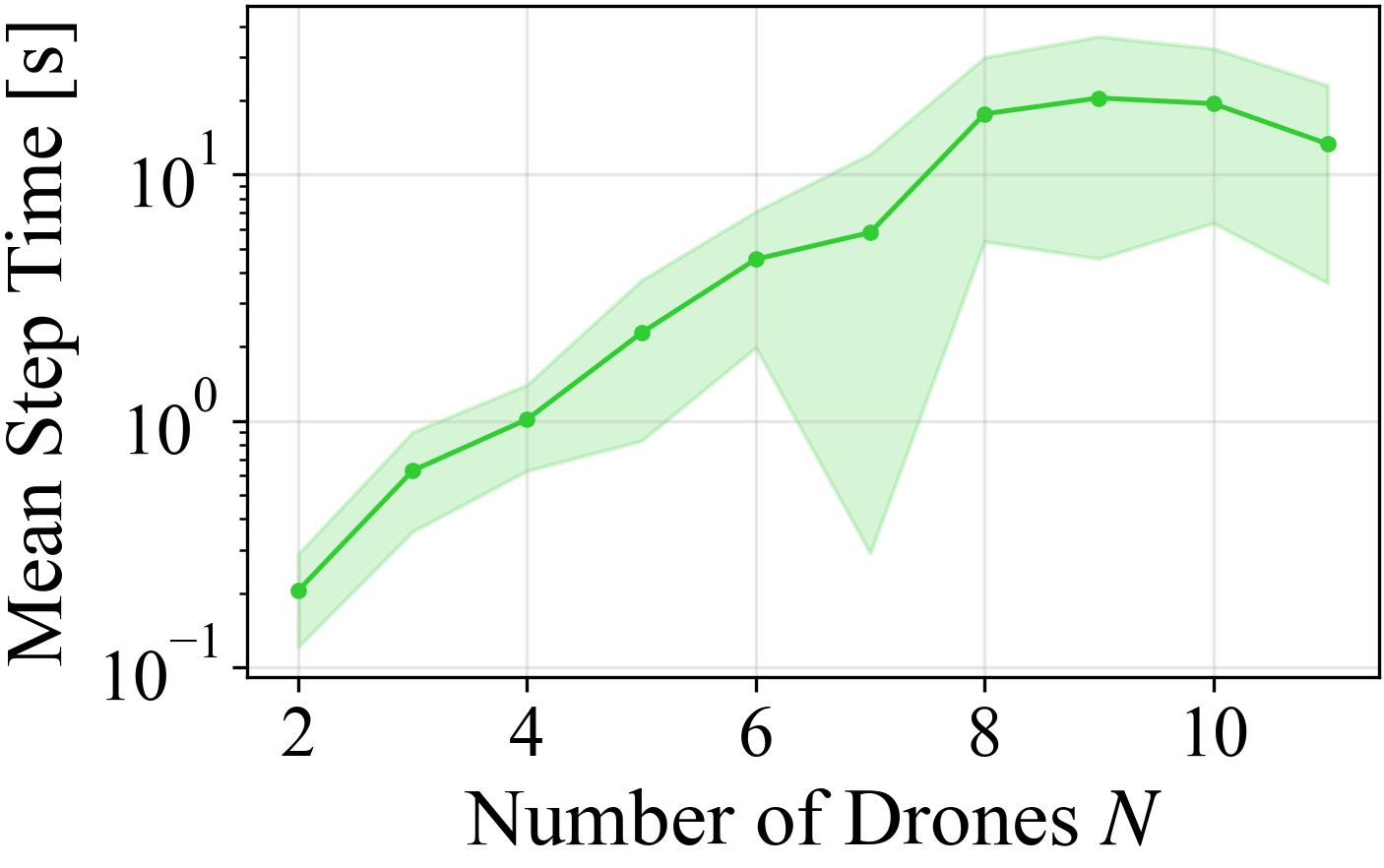}
            \caption{Computation time per step (logarithmic scale) for distributed adaptive~\cfg{B} over $N \in \{2,\ldots,11\}$.}
            \label{fig:scaling_dmpc}
        \end{minipage}
    \end{figure}

    The centralized variant~\cfg{A} (Figure~\ref{fig:scaling_mpc}) exhibits exponential growth but maintains relatively stable variance across runs, indicating consistent handling of different initial placements. In contrast, the distributed variant~\cfg{B} (Figure~\ref{fig:scaling_dmpc}) shows not only exponential growth but also sharply increasing variance with~$N$. This high variance reflects the sensitivity of the Gauss--Seidel iteration to the initial drone constellation: in unfavorable configurations, the iterative coordination requires substantially more rounds to converge, amplifying the spread of computation times.

    The per-step cost of~\cfg{B} becomes prohibitive beyond $N = 11$, causing runs to exceed the timeout before convergence.
    \cfg{A} scales more favorably but also reaches its limit beyond $N = 20$.


    \section{Discussion and Conclusion}\label{sec:conclusion}

    \subsection{Discussion of Results}

    The evaluation confirms the core theoretical predictions. All four configurations operate reliably in the low-density regime ($N \ll 10$), with comparable step counts and only moderate differences in computation time. As drone count increases, static configurations fail first: centralized static~\cfg{C} at $N = 6$, distributed static~\cfg{D} at $N = 8$, both reaching $100\,\%$ deadlock at $N = 10$ and confirming the theoretical limit $N_{\text{c}}^{\text{s}} = 10$. The mathematical safety bounds are respected throughout; no constraint violations occur in any configuration.

    Beyond $N = 10$, only adaptive configurations remain operational. The centralized solver~\cfg{A} resolves scenarios up to $N = 20$ — twice the static capacity. At the throughput-optimal operating point, $N_{\text{opt}} = 18$, the improvement factor over $N_{\text{c}}^{\text{s}} = 10$ is $1.8$. Both values lie below the theoretical maximum $N_{\text{c}}^{\text{v}} / N_{\text{c}}^{\text{s}} = 2.5\times$, which represents the stability limit under ideal conditions. This gap is not due to stability failure but to computational limitations: exponential solver time growth (Figure~\ref{fig:mean_step_time_all}) causes runs to exceed the timeout before convergence. Since not all drones fly at $v_{\text{a}}^*$ simultaneously, the centralized solver exploits heterogeneous velocity profiles, explaining why the achieved $N = 20$ exceeds the conservative planning quantity $N_{\text{opt}} = 18$.

    The distributed adaptive variant~\cfg{B} is more severely affected by computational scaling. Each drone optimizes only its own trajectory based on estimated neighbor trajectories, without access to the full global state. The sequential Gauss--Seidel updates require substantially more iterations to achieve consistency, especially in dense scenarios with strong inter-agent coupling, causing breakdown at $N = 12$. Currently, all drones are optimized within a single process; distributing computation across separate processes could alleviate this bottleneck. However, convergence is inherently slower in the distributed architecture because each agent works with approximations rather than exact neighbor trajectories. Despite this performance gap, the distributed architecture is a more realistic deployment model for mixed-fleet operations, and our stability analysis (Theorem~\ref{thm:distributed_stability}) confirms it is provably safe under the same conditions as the centralized variant.

    Throughput experiments further demonstrate the practical value of adaptive safety zones. In the exclusive passage regime, adaptive drones traverse openings geometrically impassable to static drones by reducing speed and shrinking their safety zones. In the narrow tunnel scenario, adaptive spheres reduce total traversal time by $25\,\%$. These results confirm that adaptive radii both increase collision-free drone capacity and enable passage through spatial constraints that static radii cannot accommodate. However, these benefits come at significant computational cost: nonlinear, velocity-dependent separation constraints increase per-step solver time, most pronouncedly when the MPC horizon is short, forcing a reactive slow-down strategy instead of proactive conflict avoidance.

    \subsection{Limitations and Future Work}

    The most significant limitation is the exponential computational cost increase with drone count in the adaptive formulation, especially in the distributed architecture where iterative Gauss--Seidel coordination compounds per-step overhead. One promising mitigation is improved trajectory prediction. The reactive behavior at short horizons indicates that the optimizer would benefit substantially from better estimates of future neighbor trajectories. If each drone could predict neighbor locations rather than treating last known positions as static, the optimizer could anticipate conflicts earlier and plan smoother maneuvers with fewer iterations. Learning-based methods such as LSTM networks trained on historical flight data could generate these predictions in real time, serving as warm-start information or soft constraints within the MPC formulation.

    A second strategy is a hybrid control scheme that activates the adaptive formulation only when needed. Since static and adaptive configurations perform similarly at low density, a mode-switching approach would use static radii under normal conditions and activate adaptive radii when local density exceeds a threshold — for instance, near bottlenecks or during tight last-mile delivery maneuvers. This would combine the computational efficiency of static radii with the feasibility guarantees of the adaptive approach, concentrating additional solver effort where capacity gain is needed.

    Beyond computational improvements, alternative optimization paradigms could enhance solution quality. A game-theoretic formulation in which each drone optimizes its trajectory in anticipation of others' responses could improve coordination beyond what sequential Gauss--Seidel achieves. Similarly, convergence acceleration techniques such as consensus-based warm starting or adaptive step-size strategies within the ADMM-GS decomposition could extend the distributed scheme's feasible range closer to the centralized limit.

    Finally, all results were obtained through simulation. Transferring the adaptive MPC framework to physical hardware introduces challenges including communication latency, state estimation uncertainty, and actuator limitations that may affect the feasibility bounds and stability guarantees established here.

    \subsection{Conclusion}

    This work demonstrates that velocity-dependent adaptive safety zones significantly enhance the operational capacity of MPC-based multi-drone systems in restricted airspace. The theoretical framework provides explicit, verifiable conditions for stability and feasibility under centralized and distributed control. Experimental results validate the predicted capacity increase: the centralized adaptive variant handles $1.8\times$ more drones at the throughput-optimal operating point, and up to $2\times$ more in practice. Though limited by higher computational cost, the distributed variant maintains the same mathematical safety guarantees. Throughput experiments confirm that adaptive spheres enable operations in constrained geometries infeasible with static radii. The identified trade-offs between computational overhead and airspace utilization motivate further research on trajectory prediction, hybrid control strategies, and game-theoretic coordination to bring adaptive safety zones closer to real-time deployment in multi-fleet drone operations.


    \section*{Acknowledgment}
    \input{extras/siddafunding}

    \ifCLASSOPTIONcaptionsoff
    \newpage
    \fi

    \bibliographystyle{IEEEtran}
    \bibliography{bib/bib}

    \vfill

\end{document}

%% file: extras/ntp.tex
\smallskip\noindent\textbf{Note to Practitioners:} This work provides actionable guidelines for deploying multi-drone systems in constrained environments such as warehouses or urban corridors. Given a volume~$V$, the minimum safety radius~$r_{\min}$, and the adaptation parameter~$\alpha$, the derived formulas yield the maximum drone count and the throughput-optimal operating point. For single-operator fleets, centralized adaptive MPC handles up to $1.8\times$ more drones than fixed-radius schemes; for mixed fleets with non-cooperative agents, the distributed variant optimizes locally without inter-fleet communication while maintaining identical safety guarantees. A practical hybrid strategy activates adaptive radii only when local density exceeds a threshold, preserving computational efficiency in uncongested airspace. The sufficient conditions on~$\alpha$ and drone density serve as runtime checks for operational stability.

%% file: extras/siddafunding.tex
\section*{Funding}

The research project was funded by ``The Ministry of the Environment, Nature Conservation and Transport of the State of North Rhine-Westphalia'' in Germany and co-financed by the European Union for the research project ``SIDDA - Sustainable Intermodal Drone Delivery Airline'' with grant number IN-ML-1-013b.